\documentclass{article}

\usepackage[final]{neurips_2023}

\usepackage[utf8]{inputenc} % allow utf-8 input
\usepackage[T1]{fontenc}    % use 8-bit T1 fonts
\usepackage{hyperref}       % hyperlinks
\usepackage{url}            % simple URL typesetting
\usepackage{booktabs}       % professional-quality tables
\usepackage{amsfonts}       % blackboard math symbols
\usepackage{nicefrac}       % compact symbols for 1/2, etc.
\usepackage{microtype}      % microtypography
\usepackage{xcolor}         % colors
\usepackage{tikz}
\usepackage{tikz-3dplot}
\usepackage{caption}
\usepackage{subcaption}
\usepackage{bbm}
\usepackage{multirow}
\usepackage{xfrac}

\usepackage{amsthm}
\newtheorem{proposition}{Proposition}

\newtheorem{theorem}{Theorem}
\newtheorem{lemma}{Lemma}

\newtheorem{definition}{Definition}
\newtheorem{assumption}{Assumption}
\newtheorem{notation}{Notation}

\newtheorem*{theorem*}{Theorem}

\newcommand{\pg}[1]{\paragraph{#1}\hspace{-3mm}}

\usepackage{algorithm}
\usepackage{algpseudocode}

\usepackage{amsmath,amsfonts,bm}

\def\eqref#1{equation~\ref{#1}}
\def\1{\bm{1}}

\def\rvx{{\mathbf{x}}}

\DeclareMathAlphabet{\mathsfit}{\encodingdefault}{\sfdefault}{m}{sl}
\SetMathAlphabet{\mathsfit}{bold}{\encodingdefault}{\sfdefault}{bx}{n}

\newcommand{\KL}{D_{\mathrm{KL}}}

\DeclareMathOperator*{\argmin}{arg\,min}

\usepackage{mathtools}

\DeclarePairedDelimiterX{\infdivx}[2]{[}{]}{%
  #1\delimsize\| #2%
}
\DeclarePairedDelimiterX{\infdivxcolon}[2]{[}{]}{%
  #1\delimsize: #2%
}
\newcommand{\KLD}{\KL\infdivx}
\newcommand{\infD}{D_{\infty}\infdivx}

\newcommand{\Ent}{\mathbb{H}}

\DeclarePairedDelimiterX{\innerProd}[2]{\langle}{\rangle}{%
    #1,#2%
}

\newcommand{\Oh}{\mathcal{O}}
\newcommand{\Nats}{\mathbb{N}}

\newcommand{\Reals}{\mathbb{R}}
\newcommand{\Exp}{\mathbb{E}}

\newcommand{\defeq}{\stackrel{\mathit{def}}{=}}

\newcommand{\astar}{A$^*$ }
\newcommand{\asstar}{AS$^*$ }
\newcommand{\adstar}{AD$^*$ }

\newcommand{\encode}{\mathtt{enc}}
\newcommand{\decode}{\mathtt{dec}}
\newcommand{\Prob}{\mathbb{P}}
\newcommand{\Unif}{\mathrm{Unif}}
\usepackage{makecell}
\usepackage{cleveref}
\crefname{assumption}{assumption}{assumptions}

\definecolor{orange}{HTML}{F7921D}
\definecolor{purple}{HTML}{AF72B0}
\definecolor{green}{HTML}{3C8031}
\definecolor{red}{HTML}{B6321C}

\title{Faster Relative Entropy Coding with\\ Greedy Rejection Coding}

\author{%
  Gergely Flamich\thanks{Equal contribution.}\\
  Department of Engineering\\
  University of Cambridge\\
  \texttt{gf332@cam.ac.uk}
  \And
  Stratis Markou$^*$\\
  Department of Engineering\\
  University of Cambridge\\
  \texttt{em626@cam.ac.uk}
  \And
  Jos\'e Miguel Hern\'andez Lobato \\
  Department of Engineering\\
  University of Cambridge\\
  \texttt{jmh233@cam.ac.uk}
}

\begin{document}

\maketitle

\begin{abstract}
Relative entropy coding (REC) algorithms encode a sample from a target distribution $Q$ using a proposal distribution $P$ using as few bits as possible.
Unlike entropy coding, REC does not assume discrete distributions or require quantisation.
As such, it can be naturally integrated into communication pipelines such as learnt compression and differentially private federated learning. 
Unfortunately, despite their practical benefits, REC algorithms have not seen widespread application, due to their prohibitively slow runtimes or restrictive assumptions. 
In this paper, we make progress towards addressing these issues. 
We introduce Greedy Rejection Coding (GRC), which generalises the rejection based-algorithm of \cite{harsha2007communication} to arbitrary probability spaces and partitioning schemes.
We first show that GRC terminates almost surely and returns unbiased samples from $Q$, after which we focus on two of its variants: GRCS and GRCD.
We show that for continuous $Q$ and $P$ over $\mathbb{R}$ with unimodal density ratio $dQ/dP$, the expected runtime of GRCS is upper bounded by $\beta\KLD{Q}{P} + \mathcal{O}(1)$ where $\beta \approx 4.82$, and its expected codelength is optimal.
This makes GRCS the first REC algorithm with guaranteed optimal runtime for this class of distributions, up to the multiplicative constant $\beta$.
This significantly improves upon the previous state-of-the-art method, A* coding \citep{flamich2022fast}.
Under the same assumptions, we experimentally observe and conjecture that the expected runtime and codelength of GRCD are upper bounded by $\KLD{Q}{P} + \mathcal{O}(1)$.
Finally, we evaluate GRC in a variational autoencoder-based compression pipeline on MNIST, and show that a modified ELBO and an index-compression method can further improve compression efficiency.
\end{abstract}

\vspace{-3mm}
\section{Introduction and motivation}

\vspace{-2mm}
Over the past decade, the development of excellent deep generative models (DGMs) such as variational autoencoders \citep[VAEs;][]{vahdat2020nvae,child2020very}, normalising flows \citep{kingma2016improved} and diffusion models \citep{ho2020denoising} demonstrated great promise in leveraging machine learning (ML) for data compression.
Many recent learnt compression approaches have significantly outperformed the best classical hand-crafted codecs across a range of domains including, for example, lossless and lossy compression of images and video \citep{zhang2021iflow, mentzer2020high, mentzer2022vct}.

\vspace{-3mm}
\pg{Transform coding.}
Most learnt compression algorithms are \textit{transform coding} methods: they first map a datum to a latent variable using a learnt transform, and encode it using entropy coding \citep{balle2020nonlinear}.
Entropy coding assumes discrete variables while the latent variables in DGMs are typically continuous, so most transform coding methods quantize the latent variable prior to entropy coding.
Unfortunately, quantization is a non-differentiable operation. 
Thus, state-of-the-art DGMs trained with gradient-based optimisation must resort to some continuous approximation to quantisation during training and switch to hard quantisation for compression.
Previous works have argued that using quantisation within learnt compression is restrictive or otherwise harmful, and that a method which naturally interfaces with continuous latent variables is needed \citep{havasi2018minimal,flamich2020compressing,Theis2021a,flamich2022fast}.

\vspace{-3mm}
\pg{Relative entropy coding.}
In this paper, we study \textit{relative entropy coding} \citep[REC;][]{havasi2018minimal, flamich2020compressing}, an alternative 
%communication setup which does not require quantisation.
to quantization and entropy coding.
% In particular, instead of encoding a fixed variable, REC algorithms take a target distribution $Q$ over the representation $\rvz$ as input and encode \textbf{a single sample} from $Q$.
A REC algorithm uses a proposal distribution $P$, and a public source of randomness $S$, to produce a random code which represents a \textit{single sample} from a target distribution $Q$.
Thus REC does not assume discrete distributions and interfaces naturally with continuous variables.
Remarkably, REC has fundamental advantages over quantization in lossy compression with realism constraints \citep{theis2021advantages,theis2022lossy}.
More generally, it finds application across a range of settings including, for example, differentially private compression for federated learning \citep{shah2022optimal}.

\vspace{-3mm}
\pg{Limitations of existing REC algorithms.}
While algorithms for solving REC problems already exist, most of them suffer from limitations that render them impractical.
These limitations fall into three categories: prohibitively long runtimes, overly restrictive assumptions, or additional coding overheads.
In this work, we study and make progress towards addressing these limitations.

\vspace{-3mm}
\pg{General-purpose REC algorithms.}
On the one hand, some REC algorithms make very mild assumptions and are therefore applicable in a wide range of REC problems \citep{harsha2007communication,li2018strong}.
Unfortunately, these algorithms have prohibitively long runtimes.
This is perhaps unsurprising in light of a result by \citet{agustsson2020universally}, who showed that without additional assumptions on $Q$ and $P$, the worst-case expected runtime of any general-purpose REC algorithm scales as $\smash{2^{\KLD{Q}{P}}}$, which is impractically slow.
There are also REC algorithms which accept a desired runtime as a user-specified parameter, at the expense of introducing bias in their samples \citep{havasi2018minimal,theis2021algorithms}.
Unfortunately, in order to reduce this bias to acceptable levels, these algorithms require runtimes of an order of $\smash{2^{\KLD{Q}{P}}}$, and are therefore also impractical.

\vspace{-3mm}
\pg{Faster algorithms with additional assumptions.}
On the other hand, there exist algorithms which make additional assumptions in order to achieve faster runtimes.
For example, dithered quantisation \citep{ziv1985universal, agustsson2020universally} achieves an expected runtime of $\KLD{Q}{P}$, which is optimal since any REC algorithm has an expected runtime of at least $\KLD{Q}{P}$.
However, it requires both $Q$ and $P$ to be uniform distributions, which limits its applicability.
Recently, \cite{flamich2022fast} introduced A$^*$ coding, an algorithm based on \astar sampling \citep{maddison2014sampling} which, under assumptions satisfied in practice, achieves an expected runtime of $\infD{Q}{P}$.
Unfortunately, this runtime is sub-optimal and is not always practically fast, since $\infD{Q}{P}$ can be arbitrarily large for fixed $\KLD{Q}{P}$.
Further, as discussed in \cite{flamich2022fast} this runtime also comes at a cost of an additional, substantial, overhead in codelength, which limits the applicability of A$^*$ coding.

\begin{figure}[t]
\centering
\vspace{-16mm}
    \tdplotsetmaincoords{50}{50}
    \begin{tikzpicture}[scale=2.5,tdplot_main_coords,every node/.style={scale=0.95}]
    \draw[black, fill=blue!4] (0, 0, 0) -- (1, 0, 0) -- (1, 0, 1) -- (0, 0, 1) -- cycle;
    \draw[black, fill=blue!4] (1, 0, 0) -- (1, 1, 0) -- (1, 1, 1) -- (1, 0, 1) -- cycle;
    \draw[black, fill=blue!4] (1, 0, 1) -- (1, 1, 1) -- (0, 0, 1) -- cycle;
    \node[left] (AS) at (0.0, -0.05, 0.0) {\asstar coding};
    \node[below] (Global) at (1, 0, 0) {Global \astar};
    \node[below right] (AD) at (0.95, 1.075, 0.1) {\adstar coding};
    \node[above right] (GRCD) at (0.95, 1, 1) {\color{purple} GRCD};
    \node[above left] (GRCS) at (0, 0, 1) {\color{purple} GRCS};
    \node[below right] (GRCG) at (1.0, 0, 1) {\color{purple} GRCG};
    \draw [<-, thick] (0, -0.5, 0) -- (1, -0.5, 0);
    \node[below left] (a) at (0, -0.5, 0) {Sample partitioning};
    \node[below] (b) at (1.35, -0.30, 0) {Global partitioning};
    \draw [->, thick] (1.5, 0, 0) -- (1.5, 1, 0);
    \node[below right] (d) at (1.5, 1, 0) {Dyadic partitioning};
    \draw [<->, thick] (-0.3, -0.3, 0.25) -- (-0.3, -0.3, 1);
    \node[left] (e) at (-0.3, -0.3, 0.25) {Branch \& bound search~~};
    \node[left] (f) at (-0.3, -0.3, 1.0) {Rejection coding~~};
    \end{tikzpicture}
    \vspace{-2mm}
    \caption{An illustration of the relations between the variants of GRC, introduced in this work, and the variants of A$^*$ coding. Algorithms in {\color{purple} purple} are introduced in this work. The algorithms of \cite{harsha2007communication} and \cite{li2018strong} are equivalent to GRCG and Global A$^*$ coding respectively.} 
    \label{fig:cube}
    \vspace{-5mm}
\end{figure}
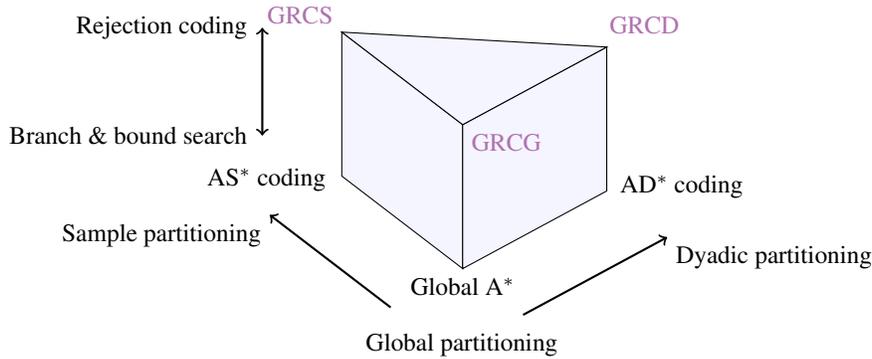
\vspace{-3mm}
\pg{Our contributions.}
In this work, we address some of these limitations.
First, we propose \textit{greedy rejection coding} (GRC), a REC algorithm based on rejection sampling.
Then, inspired by A* coding \citep{flamich2022fast}, we develop GRCS and GRCD, two variants of GRC that partition the sample space to dramatically speed up termination.
\Cref{fig:cube} illustrates the relations between GRC and its variants with existing algorithms.
We analyze the correctness and the runtime of these algorithms and, in particular, prove that GRCS has an optimal codelength and order-optimal runtime on a wide class of one-dimensional problems.
In more detail, our contributions are:
\begin{itemize}
    \item
    We introduce Greedy Rejection Coding (GRC), which generalises the algorithm of \citet{harsha2007communication} to arbitrary probability spaces and partitioning schemes.
    We prove that under mild conditions, GRC terminates almost surely and returns an unbiased sample from $Q$.
    \item
    % We show that the assumption $\infD{Q}{P} < \infty$ can be lifted, by alternative, realistic assumptions.
    % This makes GRC applicable in REC problems with $\infD{Q}{P} = \infty$, where \astar cannot be applied because, almost surely, it does not terminate.
    We introduce GRCS and GRCD, two variants of GRC for continuous distributions over $\mathbb{R}$, which adaptively partition the sample space to dramatically improve their convergence, inspired by \asstar and \adstar coding \citep{flamich2022fast}, respectively.
    % However, unlike the latter, the former can be applied even if $\infD{Q}{P} = \infty$.
    \item
    We prove that whenever $dQ / dP$ is unimodal, the expected runtime and codelength of GRCS is $\Oh(\KLD{Q}{P})$.
    This significantly improves upon the $\Oh(\infD{Q}{P})$ runtime of \asstar coding, which is always larger than that of GRCS.
    This runtime is order-optimal, while making far milder assumptions than, for example, ditered quantization.
    \item
    We provide clear experimental evidence for and conjecture that whenever $dQ / dP$ is unimodal, the expected runtime and codelength of GRCD are $\KLD{Q}{P}$.
    This also significantly improves over the $\infD{Q}{P}$ empirically observed runtime of AD$^*$ coding.
    \item
    We implement a compression pipeline with VAEs, using GRC to compress MNIST images.
    We propose a modified ELBO objective and show that this, together with a practical method for compressing the indices returned by GRC further improve compression efficiency.
\end{itemize}

\vspace{-2mm}
\section{Background and related work}
\label{sec:background}
\vspace{-2mm}
\pg{Relative entropy coding.}
First, we define REC algorithms.
% \Cref{def:rec_algorithm} formalises the definition used in \cite{flamich2022fast}, providing additional detail for completeness.
\Cref{def:rec_algorithm} is stricter than the one given by \cite{flamich2022fast}, as it has a stronger condition on the the expected codelength of the algorithm.
In this paper, all logarithms are base 2, and all divergences are measured in bits.
\begin{definition}[REC algorithm]\label{def:rec_algorithm}
    Let $(\mathcal{X}, \Sigma)$ be a measurable space, let $\mathcal{R}$ be a set of pairs of distributions $(Q, P)$ over $(\mathcal{X}, \Sigma)$ such that $\KLD{Q}{P} < \infty$ and $\mathcal{P}$ be the set of all distributions $P$ such that $(Q, P) \in \mathcal{R}$ for some distribution $Q$.
    Let $S = (S_1, S_2, \dots)$ be a publicly available sequence of independent and fair coin tosses, with corresponding probability space $(\mathcal{S}, \mathcal{F}, \mathbb{P})$ and let $\mathcal{C} = \{0, 1\}^*$ be the set of all finite binary sequences.
    A REC algorithm is a pair of functions $\encode : \mathcal{R} \times \mathcal{S} \to \mathcal{C}$ and $\decode : \mathcal{C} \times \mathcal{P} \times \mathcal{S} \to \mathcal{X}$, such that for each $(Q, P) \in \mathcal{R}$, the outputs of the encoder $C = \encode(Q, P, S)$ and the decoder $X = \decode(P, C, S)$ satisfy
    \begin{equation}
         X \sim Q \quad \text{and}\quad \mathbb{E}_S[|C|] = \KLD{Q}{P} + \mathcal{O}(\log \KLD{Q}{P}), 
    \end{equation}
    % \begin{align}
    % \label{eq:def:rec}
    % \decode(\encode(Q, P, S), P, S) \sim Q, 
    % \end{align}
    % and
    % \begin{align}
    % \mathbb{E}_S[|\encode(Q, P, S)|] = \KLD{Q}{P} + \Oh(\log(\KLD{Q}{P})),
    % \end{align}
    where $|C|$ is the length of the string $C$.
    We call $\encode$ the encoder and $\decode$ the decoder.
\end{definition}
%
% The encoder maps $(R, S)$ to $X\in \mathcal{X}$ and a binary code $C \in \mathcal{C}$, such that: (1) $C$ is uniquely decodable, (2) the expected code length is $\mathcal{O}(\KLD{Q}{P})$, and (3) $X$ is distributed according to $Q$.
In practice, $S$ is implemented with a pseudo-random number generator (PRNG) with a public seed.
In the remainder of this section, we discuss relevant REC algorithms, building up to GRC in \cref{sec:grc}.

\begin{figure}[t!]
\centering
\vspace{-20mm}
    \includegraphics[width=0.32\linewidth]{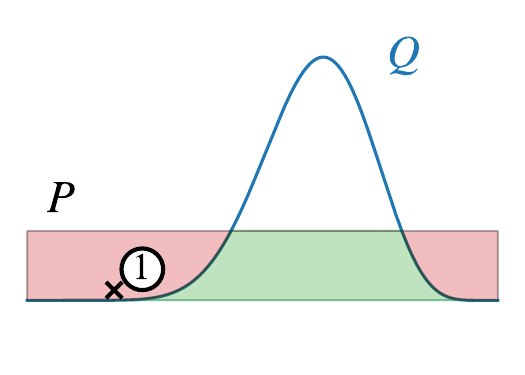}
    \includegraphics[width=0.32\linewidth]{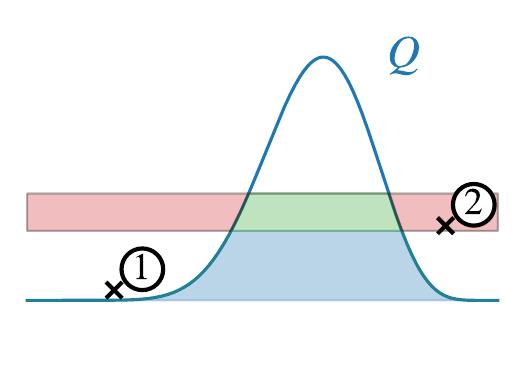}
    \includegraphics[width=0.32\linewidth]{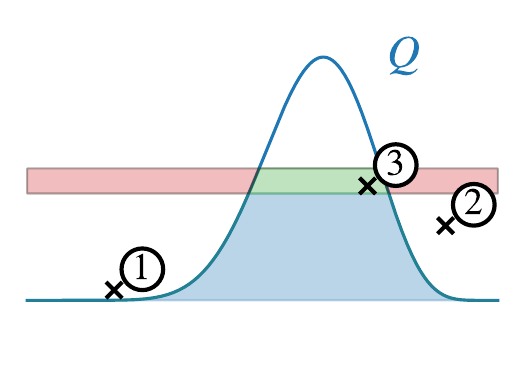}
    \vspace{-5mm}
    \caption{Example run of \cite{harsha2007communication}, for a pair of continuous $Q$ and $P$ over $[0, 1]$.
    The green and red regions correspond to acceptance and rejection regions at each step.
    Here the algorithm rejects the first two samples and accepts the third one, terminating at the third step.}
    \label{fig:grcg}
\vspace{-4mm}
\end{figure}

% \vspace{-2mm}
% \pg{Existing REC algorithms.}
% Several REC algorithms have already been developed, however they suffer from various issues, which limit their applicability in practice.
% Our proposed algorithm, Greedy Rejection Coding (GRC), is based on and generalises the rejection-based algorithm of \cite{harsha2007communication}, by drawing inspiration from A$^*$ coding \citep{flamich2022fast}.
% Specifically, A$^*$ coding can be viewed as a generalisation of an algorithm due to \cite{li2018strong}.
% The former generalises the latter by introducing a partitioning scheme to speed up termination.
% In an analogous fashion, GRC generalises \cite{harsha2007communication} by also introducing partitioning schemes, to speed up termination and achieve optimal runtimes.
% Here we discuss relevant algorithms, building up to GRC in \cref{sec:grc}.
\vspace{-2mm}
\pg{Existing REC algorithms.}
While there are many REC algorithms already, they suffer from various issues limiting their applicability in practice.
Our proposed algorithm, Greedy Rejection Coding (GRC), is based on and generalises the rejection-based algorithm of \cite{harsha2007communication}, by drawing inspiration from A$^*$ coding \citep{flamich2022fast}.
Specifically, A$^*$ coding can be viewed as a generalisation of an algorithm due to \cite{li2018strong}.
The former generalises the latter by introducing a partitioning scheme to speed up termination.
In an analogous fashion, GRC generalises \cite{harsha2007communication} by also introducing partitioning schemes, to speed up termination and achieve optimal runtimes.
Here we discuss relevant algorithms, building up to GRC in \cref{sec:grc}.
\vspace{-2mm}
\pg{REC with rejection sampling.}
\cite{harsha2007communication} introduced a REC algorithm based on rejection sampling, which we generalise and extend in this work.
While this algorithm was originally presented for discrete $Q$ and $P$, we will show that it can be generalised to arbitrary probability spaces.
In this section, we present this generalised version and in \cref{sec:grc} we further extend it to arbitrary partitioning schemes (see \cref{def:grc}).
The generalisation to arbitrary probability spaces relies on the Radon-Nikodym derivative $dQ/dP$, which is guaranteed to exist since $Q \ll P$ by \cref{def:rec_algorithm}.
When $Q$ and $P$ both have densities, $dQ/dP$ coincides with the density ratio.
\par
At each step, the algorithm draws a sample from $P$ and performs an accept-reject step, as illustrated in
\cref{fig:grcg}.
If it rejects the sample, it rules out part of $Q$ corresponding to the acceptance region, adjusts the proposal to account for the removed mass, and repeats until acceptance.
More formally, define $T_0$ to be the zero-measure on $(\mathcal{X}, \Sigma)$, and recursively for $d \in \Nats$, set:
{\normalfont \begin{align}
    T_{d+1}(S) &\stackrel{\text{def}}{=} T_d(S) + A_{d+1}(S), && A_{d+1}(S) \stackrel{\text{def}}{=} \int_S \alpha_{d+1}(x)\, dP(x), \label{eq:grcg:1} \\
% \end{align}}
% {\normalfont \begin{align}
    t_d(x) &\stackrel{\text{def}}{=} \frac{d T_d}{d P}(x), && ~\alpha_{d+1}(x) \stackrel{\text{def}}{=} \min\left\{ \frac{dQ}{dP}(x) - t_d(x), (1 - T_d(\mathcal{X})) \right\}, \label{eq:grcg:2} \\
% \end{align}}%
% {\normalfont \begin{align}
    X_d \sim P,&~~U_d \sim \text{Uniform}(0, 1) && ~\beta_{d+1}(x) \stackrel{\text{def}}{=} \frac{\alpha_{d+1}(x)}{1 - T_d(\mathcal{X})}, \label{eq:grcg:3}
\end{align}}%
for all $x\in \mathcal{X}, S \in \Sigma$.
The algorithm terminates at the first occurrence of $U_d \leq \beta_{d+1}(X_d)$.
The $T_d$ measure corresponds to the mass that has been ruled off up to and including the $d^{\text{th}}$ rejection: $T_1(\mathcal{X}), T_2(\mathcal{X})$ and $T_3(\mathcal{X})$ are the sums of the blue and green masses in the left, middle and right plots of \cref{fig:grcg} respectively.
The $A_d$ measure corresponds to the acceptance mass at the $d^{\text{th}}$ step: $A_1(\mathcal{X}), A_2(\mathcal{X})$ and $A_3(\mathcal{X})$ are the masses of the green regions in the left, middle and right plots of \cref{fig:grcg} respectively.
Lastly, $t_d, \alpha_d$ are the Radon-Nikodym derivatives i.e., roughly speaking, the densities, of $T_d, A_d$ with respect to $P$, and $\beta_{d+1}(X_d)$ is the probability of accepting the sample $X_d$.

\vspace{-1mm}
Here, the encoder $\encode$ amounts to keeping count of the number of rejections that occur up to the first acceptance, setting $C$ equal to this count and returning $X$ and $C$.
The decoder $\decode$ amounts to drawing $C+1$ samples from $P$, using the same seed as the encoder, and returning the last of these samples.
While this algorithm is elegantly simple and achieves optimal codelengths, \citet{flamich2023adaptive} showed its expected runtime is $\smash{2^{\infD{Q}{P}}}$, where $\infD{Q}{P} = \sup_{x \in \mathcal{X}} \log (dQ/dP)(x)$ is the R\'enyi $\infty$-divergence. Unfortunately, this is prohibitively slow in most practical cases.
\vspace{-2mm}
\pg{REC with Poisson \& Gumbel processes.}
\cite{li2018strong} introduced a REC algorithm based on Poisson processes, referred to as Poisson Functional Representation (PFR).
PFR assumes that $dQ/dP$ is bounded above, and relies on the fact that \citep{kingman1992poisson}, if $T_n$ are the ordered arrival times of a homogeneous Poisson process on $\mathbb{R}^+$ and $X_n \sim P$, then 
\begin{equation}
    \label{eq:pfr}
    N \stackrel{\text{def}}{=} \argmin_{n \in \mathbb{N}} \left\{T_n \frac{dP}{dQ}(X_n)\right\} \implies X_N \sim Q,
\end{equation}
Therefore, PFR casts the REC problem into an optimisation, or search, problem, which can be solved in finite time almost surely.
% Informally speaking, this is $T_n$ is monotonic increasing in $n$, so
% \begin{equation}
    % \text{ if }~ T_M \geq \inf_{x \in \mathcal{X}} \frac{dP}{dQ}(x) = \left[\sup_{x \in \mathcal{X}} \frac{dQ}{dP}(x)\right]^{-1} \text{ for some } M \in \mathbb{N}, ~\text{ then }~ N \leq M.
% \end{equation}
% Therefore, PFR eventually reaches a step $N$ beyond which $\smash{\min_{m \leq n}\{T_m (dP/dQ)(X_m)\}}$ cannot decrease further for any $n > N$, at which point the search terminates.
The PFR encoder draws pairs of samples $T_n, X_n$, until it solves the search problem in \cref{eq:pfr}, and returns $X = X_N, C = N - 1$.
The decoder can recover $X_N$ from $(P, C, S)$, by drawing $N$ samples from $P$, using the same random seed, and keeping the last sample.
While, like the algorithm of \cite{harsha2007communication}, PFR is elegantly simple and achieves optimal codelengths, its expected runtime is also $2^{\infD{Q}{P}}$ \citep{maddison2016poisson}.
%, where $\infD{Q}{P} = \max_{x \in \mathcal{X}} \log (dQ/dP)(x)$ is the Renyi $\infty$-divergence, which is also prohibitively large for practical appplications. 
\vspace{-2mm}
\pg{Fast REC requires additional assumptions.}
These algorithms' slow runtimes are perhaps unsurprising considering \citeauthor{agustsson2020universally}'s result, which shows under the computational hardness assumption $\mathrm{RP} \neq \mathrm{NP}$ that without making additional assumptions on $Q$ and $P$, there is no REC algorithm whose expected runtime scales \textit{polynomially} in $\KLD{Q}{P}$.
Therefore, in order achieve faster runtimes, a REC algorithm must make additional assumptions on $Q$ and $P$. %and leverage these assumptions accordingly.
% \pg{Fast REC requires additional assumptions.}
% The large runtimes of the aforementioned algorithms are perhaps unsurprising considering a result by \cite{agustsson2020universally}, who showed that any general-purpose REC algorithm which does not make additional assumptions on $P$ and $Q$ has an expected runtime of $\smash{\mathcal{X}(2^{\KLD{Q}{P}})}$.
% Therefore, in order achieve faster runtimes, a REC algorithm must make additional assumptions on $Q$ and $P$ and leverage these assumptions accordingly.
\vspace{-2mm}
\pg{A$^*$ coding.}
To this end, \cite{flamich2022fast} proposed: (1) a set of appropriate assumptions which are satisfied by many deep latent variable models in practice and (2) a REC algorithm, referred to as \astar coding, which leverages these assumptions to achieve a substantial speed-up over existing methods.
In particular, \astar coding generalizes PFR by introducing a partitioning scheme, which splits the sample space $\mathcal{X}$ in nested partitioning subsets, to speed up the solution of \cref{eq:pfr}.
Drawing inspiration from this, our proposed algorithm generalises \cref{eq:grcg:1,eq:grcg:2,eq:grcg:3} in an analogous manner (see \cref{fig:cube}), introducing partitioning processes (\cref{def:pp}) to speed up the algorithm's termination.
\begin{definition}[Partitioning process]
    \label{def:pp}
    A partitioning process is a process {\normalfont $Z: \mathbb{N}^+ \to \Sigma$} such that
    {\normalfont \begin{equation}
        Z_1 = \mathcal{X}, ~~ Z_{2n} \cap Z_{2n + 1} = \emptyset, ~~ Z_{2n} \cup Z_{2n + 1} = Z_n.
    \end{equation}}
\end{definition}
In other words, a partitioning process $Z$ is a process indexed by the heap indices of an infinite binary tree, where the root node is $\mathcal{X}$ and any two children nodes $Z_{2n}, Z_{2n + 1}$ partition their parent node $Z_n$.
In \cref{sec:grc} we present specific choices of partitioning processes which dramatically speed up GRC.

\vspace{-2mm}
\pg{Greedy Poisson Rejection Sampling.}
Contemporary to our work, \citet{flamich2023greedy} introduces a rejection sampler based on Poisson processes, which can be used as a REC algorithm referred to as Greedy Poisson Rejection Sampling (GPRS). 
Similar to GRC and A* coding, GPRS partitions the sample space to speed up the convergence to the accepted sample. 
Furthermore, a variant of GPRS also achieves order-optimal runtime for one-dimensional distribution pairs with a unimodal density ratio. 
However, the construction of their method is significantly different from ours, relying entirely on Poisson processes.
Moreover, GPRS requires numerically solving a certain ODE, while our method does not, making it potentially more favourable in practice. 
We believe establishing a closer connection between GPRS and GRC is a promising future research direction.
\section{Greedy Rejection Coding}
\label{sec:grc}

\vspace{-2mm}
\pg{Generalising \cite{harsha2007communication}.}
In this section we introduce Greedy Rejection Coding (GRC; \cref{def:grc}), which generalises the algorithm of \cite{harsha2007communication} in two ways.
First, GRC can be used with distributions over arbitrary probability spaces.
Therefore, it is applicable to arbitrary REC problems, including REC with continuous distributions.
Second, similar to A$^*$ coding, GRC can be combined with arbitrary partitioning processes, allowing it to achieve optimal runtimes given additional assumptions on the REC problem, and an appropriate choice of partitioning process.

\begin{figure}
    \vspace{-18mm}
    \centering
    \begin{subfigure}[b]{0.32\textwidth}
        \includegraphics[width=\linewidth]{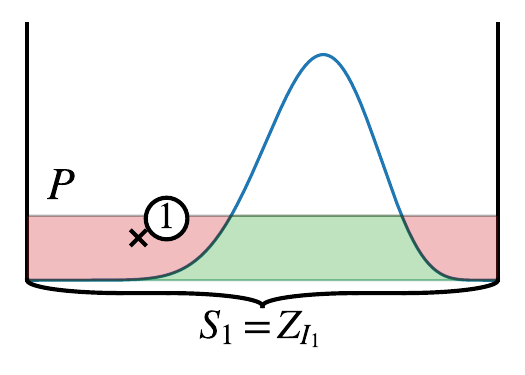}
        \subcaption{Sample \& {\color{green} accept} or {\color{red} reject}}
        \label{fig:grcd:illustration:a}
    \end{subfigure}
    \begin{subfigure}[b]{0.32\textwidth}
        \includegraphics[width=\linewidth]{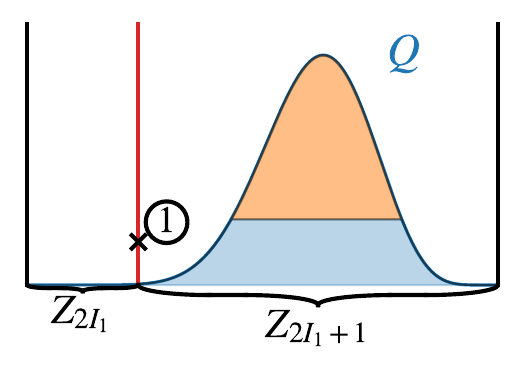}
        \subcaption{Partition \& sample $b_1 \in\{{\color{purple}\mathbf{0}}, {\color{orange}\mathbf{1}}\}$}
    \end{subfigure}
    \begin{subfigure}[b]{0.32\textwidth}
        \includegraphics[width=\linewidth]{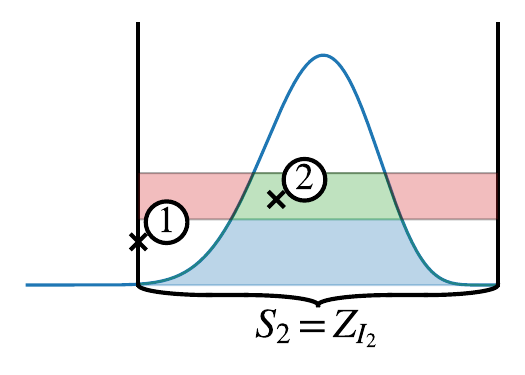}
        \subcaption{Sample \& {\color{green} accept} or {\color{red} reject}}
    \end{subfigure}
    \newline
    \begin{subfigure}[b]{0.32\textwidth}
        \includegraphics[width=\linewidth]{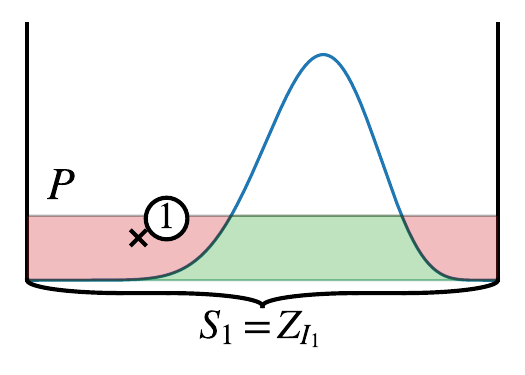}
        \subcaption{Sample \& {\color{green} accept} or {\color{red} reject}}
        \label{fig:grcd:illustration:d}
    \end{subfigure}
    \begin{subfigure}[b]{0.32\textwidth}
        \includegraphics[width=\linewidth]{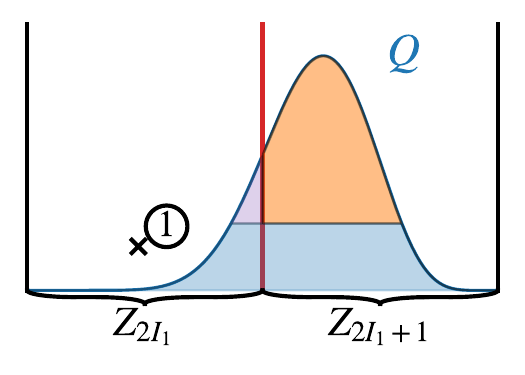}
        \subcaption{Partition \& sample $b_1 \in\{{\color{purple}\mathbf{0}}, {\color{orange}\mathbf{1}}\}$}
        \label{fig:grcd:illustration:e}
    \end{subfigure}
    \begin{subfigure}[b]{0.32\textwidth}
        \includegraphics[width=\linewidth]{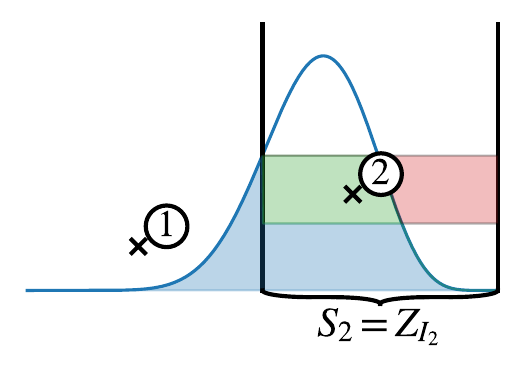}
        \subcaption{Sample \& {\color{green} accept} or {\color{red} reject}}
        \label{fig:grcd:illustration:f}
    \end{subfigure}
    \caption{
    Illustrations of the two variants of GRC considered in this work.
    (a) to (c) show GRC with the \textit{sample-splitting} partitioning process (GRCS).
    (d) to (f) show GRC with the dyadic partition process (GRCD).
    GRC interleaves accept-reject steps with partitioning steps.
    In the former, it draws a sample and either accepts or rejects it.
    In the latter, it partitions the sample space and randomly chooses one of the partitions, ruling out large parts of the sample space and speeding up termination.
    }
    \label{fig:grcd:illustration}
    \vspace{-4mm}
\end{figure}

\subsection{Algorithm definition}
\vspace{-1mm}
\pg{Overview.}
Before specifying GRC, we summarise its operation with an accompanying illustration.
On a high level, GRC interleaves accept-reject steps with partitioning steps, where the latter are determined by a partitioning process.
Specifically, consider the example in \cref{fig:grcd:illustration:d,fig:grcd:illustration:e,fig:grcd:illustration:f}, where $Q$ and $P$ are distributions over $\mathcal{X} = [0, 1]$, and $Z$ is the partitioning process defined by
\begin{equation}
    Z_n = [L, R] \implies Z_{2n} = [L, M), Z_{2n+1} = [M, R], \text{ where } M = (L + R) / 2.
\end{equation}
In each step $d = 1, 2, \dots$, GRC maintains a heap index $I_d$ of an infinite binary tree, and an active subset $S_d = Z_{I_d} \subseteq \mathcal{X}$ of the sample space, initialised as $I_0 = 1$ and $S_1 = Z_1 = \mathcal{X}$ respectively.

\vspace{-2mm}
\pg{Accept-reject step.}
In each step, GRC draws a sample from the restriction of $P$ to $S_d$, namely $P|_{S_d} / P(S_d)$, and either accepts or rejects it.
If the sample is accepted, the algorithm terminates.
Otherwise, GRC performs a partitioning step as shown in \cref{fig:grcd:illustration:d}

\vspace{-2mm}
\pg{Partitioning step.}
In each partitioning step, GRC partitions $S_d = Z_{I_d}$ into $Z_{2I_d}$ and $Z_{2I_d + 1}$, as specified by the partitioning process $Z$.
It then samples a Bernoulli random variable $b_d$, whose outcomes have probabilities proportional to the mass of $Q$ which has not been accounted for, up to and including step $d$, within the partitions $Z_{2I_d}$ and $Z_{2I_d + 1}$ respectively.
In \cref{fig:grcd:illustration:e}, these two masses correspond to the purple and orange areas, and the algorithm has sampled $b_d = 1$.
Last, GRC updates the heap index to $I_{d+1} = 2I_d + b_d$ and the active subset to $S_{d+1} = Z_{I_{d+1}}$.
GRC proceeds by interleaving accept-reject and partitioning steps until an acceptance occurs.

\vspace{-2mm}
\pg{Algorithm specification.}
The aforementioned algorithm can be formalised in terms of probability measures over arbitrary spaces and arbitrary partitioning processes.
Above, \cref{alg:harsha,alg:grc} describe \citeauthor{harsha2007communication}'s rejection sampler and our generalisation of it, respectively.
For the sake of keeping the exposition lightweight, we defer the formal measure-theoretic definition of GRC to the appendix (see \cref{def:grc} in \cref{app:def:grc}), and refer to \cref{alg:grc} as a working definition here.

\begin{figure}[t]
\vspace{-16mm}
\begin{minipage}{0.48\textwidth}
\begin{algorithm}[H]
    \centering
    \caption{\citeauthor{harsha2007communication}'s rejection algorithm; equivalent to GRC with a global partition}
    \footnotesize
    \begin{algorithmic}[1]
        \Require Target $Q$, proposal $P$, space $\mathcal{X}$
        \State $d \gets 0, T_0 \gets 0$
        \State
        \While{\texttt{True}}
            \State $X_{d+1} \sim P$
            \State $U_{d+1} \sim \text{Uniform}(0, 1)$
            \State $\beta_{d+1} \gets \texttt{AcceptProb}(Q, P, X_{d+1}, T_d)$
                \If{$U_{d+1} \leq \beta_{d+1}$}
                    \State \textbf{return} $X_{d+1}, d$
                \EndIf
            \State 
            \State 
            \State 
            \State $T_{d+1} \gets \texttt{RuledOutMass}(Q, P, T_d)$
            \State $d \gets d + 1$
        \EndWhile
    \end{algorithmic}
    \label{alg:harsha}
\end{algorithm}
\end{minipage}
\hfill
\begin{minipage}{0.50\textwidth}
\begin{algorithm}[H]
    \centering
    \caption{GRC with partition process $Z$; differences to \citeauthor{harsha2007communication}'s algorithm shown in {\color{green} green}}
    \footnotesize
    \begin{algorithmic}[1]
        \Require Target $Q$, proposal $P$, space $\mathcal{X}$, {\color{green}partition $Z$}
        \State $d \gets 0, T_0 \gets 0$
        \State ${\color{green} I_0 \gets 1, S_1 \gets \mathcal{X}}$
        % \While{{\color{green}$d \leq D_{\max}$}}
        \While{$\mathtt{True}$}
            \State $X_{I_d} \sim {\color{green} P|_{S_d}/P(S_d)}$
            \State $U_{I_d} \sim \text{Uniform}(0, 1)$
            \State $\beta_{I_d} \gets \texttt{AcceptProb}(Q, P, X_{I_d}, T_d)$
                \If{$U_{I_d} \leq \beta_{d+1}$ 
                {\color{green}\texttt{or} $d = D_{\max}$}
                }
                    \State \textbf{return} $X_{I_d}, I_d$
                \EndIf
            \State {\color{green}$p \gets \texttt{PartitionProb}(Q, P, T_d, Z_{2d}, Z_{2d+1})$}
            \State {\color{green} $b_d \sim \text{Bernoulli}(p)$}
            \State {\color{green} $I_{d+1} \gets 2I_d + b_d$ 
            \text{ and }$S_{d+1} \gets Z_{I_{d+1}}$}
            \State $T_{d+1} \gets \texttt{RuledOutMass}(Q, P, T_d, {\color{green} S_{d+1}})$
            \State $d \gets d + 1$
        \EndWhile
    \end{algorithmic}
    \label{alg:grc}
\end{algorithm}
\end{minipage}
\vspace{-3mm}
\end{figure}

\vspace{-2mm}
\pg{Comparison to \citeauthor{harsha2007communication}}
While \cref{alg:harsha,alg:grc} are similar, they differ in two notable ways.
First, rather than drawing a sample from $P$, GRC draws a sample from the restriction of $P$ to an active subset $S_d = Z_d \subseteq \mathcal{X}$, namely $P|_{S_d}/P(S_d)$.
Second, GRC updates its active subset $S_d = Z_d$ at each step, setting it to one of the children of $Z_d$, namely either $Z_{2d}$ or $Z_{2d+1}$, by drawing $b_d \sim \text{Bernoulli}$, and setting $Z_{2d+b_d}$.
This partitioning mechanism, which does not appear in \cref{alg:harsha}, yields a different variant of GRC for each choice of partitioning process $Z$.
In fact, as shown in \Cref{proposition:grcg} below, \cref{alg:harsha} is a special case of GRC with $S_d = \mathcal{X}$ for all $d$.
See \cref{app:harsha:special} for the proof.
\begin{proposition}[\cite{harsha2007communication} is a special case of GRC]
\label{proposition:grcg}
    Let $Z$ be the global partitioning process over $\Sigma$, defined as
    \begin{equation}
        \label{eq:global:Z}
        Z_1 = \mathcal{X}, ~~~ Z_{2n} = Z_n, ~~~Z_{2n+1} = \emptyset, ~~\text{ for all }~ n = 1, 2, \dots.
    \end{equation}
    \cite{harsha2007communication} is equivalent to GRC using this $Z$ and setting $C = D^*$ instead of $C = I_{D^*}$.
    We refer to this algorithm as Global GRC, or \textbf{GRCG} for short.
\end{proposition}
\vspace{-4mm}
\pg{Partitioning processes and additional assumptions.}
While \Cref{proposition:grcg} shows that \citeauthor{harsha2007communication}'s algorithm is equivalent to GRC with a particular choice of $Z$, a range of other choices of $Z$ is possible, and this is where we can leverage additional structure.
% However, as \cite{agustsson2020universally} showed, without additional assumptions on $\mathcal{R}$, any REC algorithm has a worst-case expected runtime of $\smash{\mathcal{X}(2^{\KLD{Q}{P}})}$.
% Therefore, in order to achieve faster runtimes, it is necessary to make additional assumptions on $\mathcal{R}$ and choose a partitioning process which leverages these assumptions.
In particular, we show that when $Q$ and $P$ are continuous distributions over $\Reals$ with a unimodal density ratio $dQ/dP$, we can dramatically speed up GRC with an appropriate choice of $Z$.
In particular, we will consider the sample-splitting and dyadic partitioning processes from \cite{flamich2022fast}, given in \Cref{def:sample:Z,def:dyadic:Z}.
% In particular, we show that assuming $Q$ and $P$ are continuous distributions over $\mathbb{R}$, with a unimodal density ratio $dQ/dP$, and choosing an appropriate $Z$, allows GRC to achieve a dramatic runtime speedup.
% In particular, we will consider the sample-splitting and dyadic partitioning processes from \cite{flamich2022fast}, given in \Cref{def:sample:Z,def:dyadic:Z}.
\vspace{2mm}
\begin{definition}[Sample-splitting partitioning process]\label{def:sample:Z}
    Let $\mathcal{X} = \mathbb{R} \cup \{-\infty, \infty\}$ and $P$ a continuous distribution.
    The \textit{sample-splitting} partitioning process is defined as
    \begin{equation*}
        Z_n = [a, b], a, b \in \mathcal{X} \implies Z_{2n} = [a, X_n],~~Z_{2n+1} = [X_n, b], \text{ where } X_n \sim P|_{Z_n}/P(Z_n).
    \end{equation*}
\end{definition}
In other words, in the sample-splitting process, $Z_n$ are intervals of $\mathbb{R}$, each of which is partitioned into sub-intervals $Z_{2n}$ and $Z_{2n+1}$ by splitting at the sample $X_n$ drawn from $P|_{Z_n}/P(Z_n)$.
We refer to GRC with the sample-splitting partitioning process as \textbf{GRCS}.
\vspace{2mm}
\begin{definition}[Dyadic partitioning process]\label{def:dyadic:Z}
    % Let $\mathcal{X} = \mathbb{R} \cup \{-\infty, \infty\}$ and $Q$ and $P$ be continuous distributions with $Q \ll P$.
    Let $\mathcal{X} = \mathbb{R} \cup \{-\infty, \infty\}$ and $P$ a continuous distribution.
    The \textit{dyadic} partitioning process is defined as
    \begin{equation*}
        Z_n = [a, b], a, b \in \mathcal{X} \implies Z_{2n} = [a, c],~~Z_{2n+1} = [c, b], \text{ such that } P(Z_{2n}) = P(Z_{2n+1}).
    \end{equation*}
\end{definition}
Similar to the sample-splitting process, in the dyadic process $Z_n$ are intervals of $\mathbb{R}$.
However, in the dyadic process, $Z_n$ is partitioned into sub-intervals $Z_{2n}$ and $Z_{2n+1}$ such that $P(Z_{2n}) = P(Z_{2n+1})$.
We refer to GRC with the dyadic partitioning process as \textbf{GRCD}.
\vspace{-2mm}
\pg{GRC with a tunable codelength.}
\citeauthor{flamich2022fast} presented a depth-limited variant of \adstar coding, DAD$^*$ coding, in which the codelength $|C|$ can be provided as a tunable input to the algorithm.
Fixed-codelength REC algorithms are typically approximate because they introduce bias in their samples, but are nevertheless useful in certain contexts, such as for coding a group of random variables with the same fixed codelength.
GRCD can be similarly modified to accept $|C|$ as an input, by limiting the maximum steps of the algorithm by $D_{\max}$ (see \cref{alg:grc}).
Setting $D_{\max} = \infty$ in \cref{alg:grc} corresponds to exact GRC, while setting $D_{\max} < \infty$ corresponds to depth-limited GRC.

\subsection{Theoretical results}

\vspace{-1mm}
\pg{Correctness of GRC.}
% A priori, it is unclear whether GRC is guaranteed to terminate and whether it yields exact samples from $Q$.
In \cref{thm:correctness} we show that GRC terminates almost surely and produces unbiased samples from $Q$, given interchangeable mild assumptions on $Q, P$ and $Z$.
\Cref{assumption:finite_ratio_mode} is the most general, since it holds for any $Q$ and $P$ over arbitrary probability spaces, and can be used to apply GRC to arbitrary coding settings.
\begin{assumption}
    \label{assumption:finite_ratio_mode}
    GRC has a finite ratio mode if $dQ/dP(x) < M$ for all $x \in \mathcal{X}$, for some $M \in \mathbb{R}$.
\end{assumption}
\Cref{assumption:finite_ratio_mode} holds for GRCG, GRCS and GRCD, so long as $dQ/dP$ is bounded.
While this assumption is very general, in some cases we may want to consider $Q, P$ with unbounded $dQ/dP$.
To this end, we show that it can be replaced by alternative assumptions, such as \cref{assumption:single_branch,assumption:nice_shrinking}.
\begin{assumption}
    \label{assumption:single_branch}
    GRC is single-branch if for each $d$, $b_d = 0$ or $b_d = 1$ almost surely.
\end{assumption}
\vspace{-1mm}
GRC with the global partitioning process (eq. \ref{eq:global:Z}) satisfies \cref{assumption:single_branch}.
In addition, if $Q$ and $P$ are distributions over $\mathbb{R}$ and $dQ/dP$ is unimodal, GRCS also satisfies \cref{assumption:single_branch}.
\begin{assumption}
    \label{assumption:nice_shrinking}
    Suppose $\mathcal{X} \subseteq \mathbb{R}^N$.
    GRC has nicely shrinking $Z$ if, almost surely, the following holds.
    For each $x \in \mathcal{X}$ which is in a nested sequence of partitions $x \in Z_1 \supseteq \dots \supseteq Z_{k_d} \supseteq \dots$ with $P(Z_{k_d}) \to 0$, there exist $\gamma, r_1, r_2, ... \in \mathbb{R}_{>0}$ such that
    \begin{equation}
        r_d \to 0,~ Z_{k_d} \subseteq B_{r_d}(x) \text{ and } P(Z_{k_d}) \geq \gamma P(B_{r_d}(x)).
    \end{equation}
\end{assumption}
\vspace{-2mm}
If $Q$ and $P$ are distributions over $\mathbb{R}$, GRCD satisfies assumption \ref{assumption:nice_shrinking}.
\Cref{thm:correctness} shows that if any of the above assumptions hold, then GRC terminates almost surely and yields unbiased samples from $Q$.
We provide the proof of the theorem in \cref{app:correctness}.
\begin{theorem}[Correctness of GRC]
    \label{thm:correctness}
    Suppose $Q, P$ and $Z$ satisfy any one of \cref{assumption:finite_ratio_mode,assumption:single_branch,assumption:nice_shrinking}.
    Then, \cref{alg:grc} terminates with probability $1$, and its returned sample $X$ has law $X \sim Q$.
    % Let $\tau$ be the distribution of samples produced by GRC.
    % Then GRC terminates almost surely and 
    % \begin{equation}
    %     ||Q - \tau||_{TV} = 0.
    % \end{equation}
\end{theorem}
\vspace{-5mm}
\pg{Expected runtime and codelength of GRCS.}
Now we turn to the expected runtime and codelength of GRCS.
\Cref{thm:grcs:codelength} shows that the expected codelength of GRCS is optimal, while \Cref{thm:grcs:runtime} establishes that its runtime is order-optimal.
We present the proofs of the theorems in \cref{app:grcs_proofs}.
\vspace{2mm}
\begin{theorem}[GRCS codelength]
\label{thm:grcs:codelength}
    Let $Q$ and $P$ be continuous distributions over $\mathbb{R}$ such that $Q \ll P$ and with unimodal $dQ/dP$. Let $Z$ be the sample-splitting process, and $X$ its returned sample.
    Then,
    \begin{equation} \label{eq:codelength}
    \Ent[X | Z] \leq \KLD{Q}{P} + 2\log\left(\KLD{Q}{P} + 1\right) + \Oh(1).
    \end{equation}
\end{theorem}
\begin{theorem}[GRCS runtime]
\label{thm:grcs:runtime}
    Let $Q$ and $P$ be continuous distributions over $\mathbb{R}$ such that $Q \ll P$ and with unimodal $dQ/dP$. Let $Z$ be the sample-splitting process and $D$ the number of steps the algorithm takes before accepting a sample.
    Then, for $\beta = 2/\log(4/3) \approx 4.82$ we have
    \begin{equation} \label{eq:runtime}
        \mathbb{E}[D] \leq \beta~\KLD{Q}{P} + \Oh(1)
    \end{equation}
\end{theorem}
\pg{Improving the codelength of GRCD.}
In \Cref{thm:grcs:codelength} we state the bound for the REC setting, where we make no further assumptions on $Q$ and $P$.
However, we can improve the bound if we consider the \textit{reverse channel coding} (RCC) setting \citep{theis2021algorithms}.
In RCC, we have a pair of correlated random random variables $X, Y \sim P_{X, Y}$.
During one round of communication, the encoder receives $Y \sim P_Y$ and needs to encode a sample $X \sim P_{X | Y}$ from the posterior using $P_X$ as the proposal distribution.
Thus, RCC can be thought of as the average-case version of REC, where the encoder sets $Q \gets P_{X | Y}$ and $P \gets P_X$.
In this case, when the conditions of \Cref{thm:grcs:codelength} hold for every $(P_{X | Y}, P_X)$ pair, in \cref{app:grcs_proofs} we show that the bound can be improved to $\mathbb{I}[X; Y] + 2\log(\mathbb{I}[X; Y] + 1) + \Oh(1)$, where $\mathbb{I}[X; Y] = \Exp_{Y \sim P_Y}\left[\KLD{P_{X | Y}}{P_Y}\right]$ is the mutual information between $X$ and $Y$.
% \begin{theorem}[GRCS runtime]
% \label{thm:grcs:runtime}
%     Let $Q$ and $P$ be continuous distributions over $\mathbb{R}$ with $Q \ll P$ and unimodal $dQ/dP$, and let $Z$ be the sample-splitting partitioning process.
%     Then
%     \begin{equation} \label{eq:runtime}
%         \mathbb{E}[D^*] \leq \beta~\KLD{Q}{P} + \beta~\log e.
%     \end{equation}
%     where $\beta = 2\log(4/3)$. In addition, the expected codelength of GRCS is upper-bounded as
%     \begin{equation} \label{eq:codelength}
%         \mathbb{E}[|C|] \leq \beta~\KLD{Q}{P} + \log\left(\KLD{Q}{P} + \log e\right) + \log \beta + 1,
%     \end{equation}
%     which, using an arithmetic coding scheme to further compress $C$, can be improved to 
%     \begin{equation} \label{eq:codelength:withac}
%         \mathbb{E}[|C|] \leq \KLD{Q}{P} + \log\left(\KLD{Q}{P} + \log e\right)  + \log \beta + 1.
%     \end{equation}
% \end{theorem}
\vspace{-2mm}
\pg{GRCS runtime is order-optimal.}
\Cref{thm:grcs:runtime} substantially improves upon the runtime of \astar coding, which is the current fastest REC algorithm with similar assumptions.
In particular, \asstar coding has $\mathcal{O}(\infD{Q}{P})$ expected runtime, which can be arbitrarily larger than that of GRCS.
Remarkably, the runtime of GRCS is optimal up to the multiplicative factor $\beta$.
This term arises from the fact the sample-splitting process may occasionally rule out a small part of the sample space at a given step.
% \pg{GRCS codelength.}
% While, due to the leading $\beta$ factor, the codelength upper bound in \cref{eq:codelength} is not tight enough to be of practical use, it is possible to eliminate it by applying an Arithmetic Coding scheme to the codes produced by GRCS, arriving at the tighter bound in \cref{eq:codelength:withac} which is tight enough for practical purposes (see \cref{app:theory} for details).
% This scheme however does not eliminate the $\log \beta$ additive term, which arises from the sample-splitting process.
% Motivated by this, in \cref{sec:synthetic_experiments} we examine the runtime of GRCD, which always rules out half of the sample space, as weighted by $P$.
% We provide empirical evidence and conjecture that, for $Q, P$ over $\mathbb{R}$ with unimodal $dQ/dP$, GRCD achieves an optimal expected runtime and codelength given by \cref{eq:runtime,eq:codelength} with $\beta = 1$.
\vspace{-4mm}
\section{Experiments}
\vspace{-2mm}
We conducted two sets of experiments: one on controlled synthetic REC problems to check the predictions of our theorems numerically, and another using VAEs trained on MNIST to study how the performance of GRC-based compression pipelines can be improved in practice.
We conducted all our experiments under fair and reproducible conditions and make our source code public.\footnote{Source code to be published with the camera-ready version: \texttt{https://github.com/source-code}.}
\vspace{-1mm}
\subsection{Synthetic Experiments}
\label{sec:synthetic_experiments}
\vspace{-1mm}
\pg{Synthetic REC experiments.}
First, we compare GRCS and GRCD, against AS$^*$ and \adstar coding, on a range of synthetic REC problems.
We systematically vary distribution parameters to adjust the difficulty of the REC problems.
\Cref{fig:exact_rec} shows the results of our synthetic experiments.

\begin{figure}[t!]
    \vspace{-16mm}
    \centering
    \begin{subfigure}[b]{0.49\textwidth}
        \includegraphics[width=\textwidth]{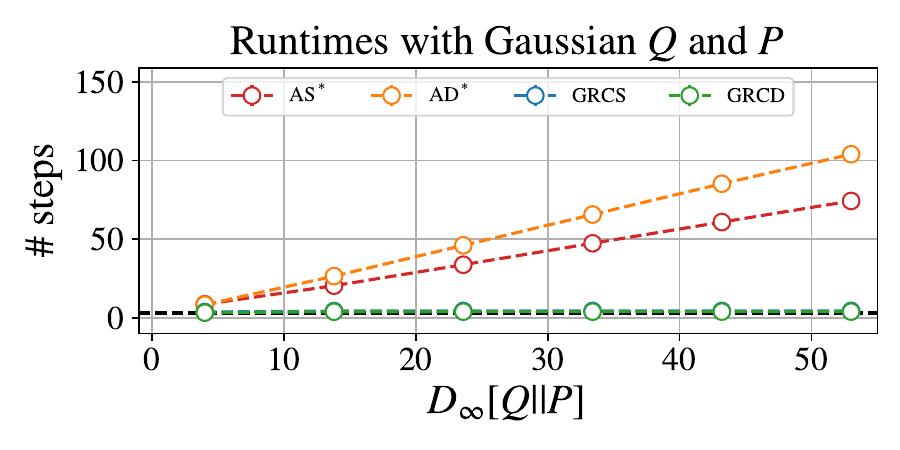}
    \end{subfigure}
    \begin{subfigure}[b]{0.49\textwidth}
        \includegraphics[width=\textwidth]{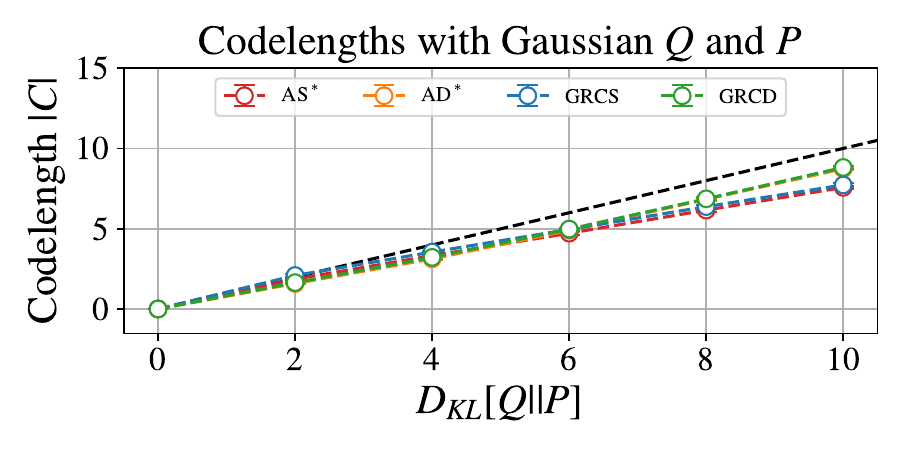}
    \end{subfigure}
    \vspace{-2mm}
    \caption{
    Comparison between GRC and \astar coding on synthetic REC problems with Gaussian $Q$ and $P$.
    \textit{Left:} we fix $\KLD{Q}{P} = 3$ and vary $\infD{Q}{P}$, measuring the number of steps taken by each algorithm.
    \textit{Right:} we fix $\infD{Q}{P} = \KLD{Q}{P} + 2$ and vary $\KLD{Q}{P}$, plotting the codelengths produced by each algorithm.
    Reported codelengths do not include additional logarithmic overhead terms.
    Results are averaged over $4 \times 10^3$ different random seeds for each datapoint.
    We have included error-bars in both plots but these are too small to see compared to the plot scales.
    }
    \label{fig:exact_rec}
    \vspace{-4mm}
\end{figure}

\vspace{-2mm}
\pg{Partitioning processes improve the runtime of GRC.}
First, we observe that, assuming that $dQ/dP$ is unimodal, introducing an appropriate partitioning process such as the sample-splitting or the dyadic process, dramatically speeds up GRC.
In particular, \cref{fig:exact_rec} shows that increasing the infinity divergence $\infD{Q}{P}$ (for a fixed $\KLD{Q}{P}$) does not affect the runtimes of GRCS and GRCD, which remain constant and small.
This is a remarkable speed-up over the exponential expected runtime of GRCG.

\vspace{-2mm}
\pg{GRC is faster than \astar coding.}
Further, we observe that GRC significantly improves upon the runtime of A* coding, which is the fastest previously known algorithm with similar assumptions.
In particular, \Cref{fig:exact_rec} shows that increasing the infinity divergence $\infD{Q}{P}$, while keeping the KL divergence $\KLD{Q}{P}$ fixed, increases the runtime of both \asstar and \adstar coding, while the runtimes of GRCS and GRCD remain constant.
More generally, for a fixed KL divergence, the infinity divergence can be arbitrarily large or even infinite.
In such cases, \astar coding would be impractically slow or even inapplicable, while GRCS and GRCD remain practically fast.

\vspace{-2mm}
\pg{GRCD improves on GRCS.}
In our experiments, we observe that the performance of GRCD (green in \cref{fig:exact_rec}) matches that of GRCS (blue in \cref{fig:exact_rec}) in terms of runtime and codelength.
While in our experiments, GRCD does not yield an improvement over GRCS, we note the following behaviour.
The sample-splitting process may occasionally rule out a only a small part of space, which can slow down convergence.
In particular, in \cref{app:grcs_proofs} we show that on average, the sample-splitting process rules out $\sfrac{1}{2}$ of the active sample space in the best case at each step, and $\sfrac{3}{4}$ in the worst case.
By contrast, the dyadic process always rules out $\sfrac{1}{2}$ of the sample space, potentially speeding up termination.
We conjecture that GRCD achieves an optimal expected runtime with $\beta = 1$.

\vspace{-1mm}
\subsection{Compression with Variational Autoencoders}
\label{sec:vae}

\vspace{-2mm}
\pg{Compressing images with VAEs and REC.}
One of the most promising applications of REC is in learnt compression.
Here, we implement a proof-of-concept lossless neural compression pipeline using a VAE with a factorized Gaussian posterior on MNIST and take the architecture used by \citet{townsend2018practical}.
% Once trained, we use the VAE to compress an image $Y$, by first passing it through the inference network to obtain the posterior $q(X \mid Y)$. 
% Then, we encode a latent sample $X \sim q(X \mid Y)$ into a code $C$ with GRC using the posterior $q(X \mid Y)$ as the target $Q$ and the VAE prior $p(X)$ as the proposal $P$. 
% Finally, we reconstruct the image by decoding $C$ and passing $X$ through the generative network of the VAE.
To compress an image $Y$, we encode a latent sample $X$ from the VAE posterior $q(X \mid Y)$ by applying GRCD dimensionwise after which we encode the image $Y$ with entropy coding using the VAE's conditional likelihood $p(Y \mid X)$ as the coding distribution.
Unfortunately, in addition to the $\KLD{q(X_d \mid Y)}{p(X_d)}$ bits coding cost for latent dimension $d$, this incurs an overhead of ${\log(\KLD{q(X_d \mid Y)}{p(X_d)} + 1) + \Oh(1)}$ bits, analogously to how a symbol code, like Huffman coding, incurs a constant overhead per symbol \citep[][]{mackay2003information}.
However, since $\log(1 + x) \approx x$ when $x \approx 0$, the logarithmic overhead of GRC can become significant compared to the KL divergence.
Hence, we now investigate two approaches to mitigate this issue.

\setlength\tabcolsep{4.0pt}
\begin{table}[t!] \small
\vspace{-12mm}
\begin{center}
\begin{tabular}{ c c c c c c c } 
 \toprule
 \makecell{\scshape Training \\ \scshape objective} & {\scshape \# latent} & \makecell{\scshape Total BPP \\ \scshape with $\zeta$ coding} & \makecell{\scshape Total BPP \\ \scshape with $\delta$ coding} & \makecell{\scshape Neg. ELBO \\\scshape per pixel} & \makecell{\scshape Overhead BPP \\ \scshape with $\delta$ coding}\\ \midrule
 \multirow{ 3}{*}{{\scshape ELBO}} & 20 & { $1.472 \pm 0.004$} & { $1.482 \pm 0.004$} & $1.391 \pm 0.004$ & $0.091 \pm 0.000$ \\
  & 50 & { $1.511 \pm 0.003$} & { $1.530 \pm 0.003$} & $1.357 \pm 0.003$ & $0.172 \pm 0.000$ \\ 
  & 100 & { $1.523 \pm 0.003$} & { $1.600 \pm 0.003$} & $1.362 \pm 0.003$ & $0.238 \pm 0.000$ \\ \midrule
 \multirow{ 3}{*}{{\scshape Modified ELBO}} & 20 & { $1.470 \pm 0.004$} & { $1.478 \pm 0.004$} & { $1.393 \pm 0.004$} & $0.085 \pm 0.000$\\
  & 50 & { $1.484 \pm 0.003$} & { $1.514 \pm 0.003$} & { $1.373 \pm 0.003$} & $0.141 \pm 0.000$ \\
  & 100 & { $1.485 \pm 0.003$} & { $1.579 \pm 0.003$} & { $1.373 \pm 0.003$} & $0.205 \pm 0.000$ \\
 \bottomrule
\end{tabular}
\end{center}
\caption{Lossless compression performance comparison on the MNIST test set of a small VAE with different latent space sizes, optimized using either the ELBO or the modified ELBO in \cref{eq:modified_elbo}.
We report the bits per pixel (BPP) attained using different coding methods, averaged over the 10,000 test images, along with the standard error, using GRCD. 
See \cref{sec:vae} for further details.}
\label{tab:mnist1}
\vspace{-8mm}
\end{table}

\pg{Modified ELBO for REC.} 
A principled approach to optimizing our neural compression pipeline is to minimize its expected codelength.
% Fortunately, this lower bound is achievable within a constant overhead using entropy coding. 
% Hence minimizing the negative ELBO is also equivalent to minimizing an \textit{upper bound} on the expected coding cost of a lossless compressor.
% However, as we discuss above, this is no longer true in the case of relative entropy coding, as the upper bound has an additional logarithmic overhead.
% Previous works, such as \citep{flamich2022fast}, trained their VAEs for REC by minimizing the negative ELBO. However, this is no longer equivalent to minimizing an upper bound on the coding cost.
% \citet{townsend2018practical} point out that the negative ELBO forms an \textit{upper bound} to the entropy of the marginal data distribution.
For bits-back methods \citep{townsend2018practical, townsend2019hilloc},
the negative ELBO indeed expresses their expected codelength, but in REC's case, it does not take into account the additional dimensionwise logarithmic overhead we discussed above.
Thus, we propose to minimize a modified negative ELBO to account for this (assuming that we have $D$ latent dimensions):
\vspace{-4mm}
\begin{align}
    % \Exp_{\rvx \sim p(\rvx)}\Bigg[ 
    \underbrace{\Exp_{X \sim q(X | Y)}[-\log p(Y | X)] + \KLD{q(X | Y)}{p(X)}}_{\text{Regular ELBO}} + \sum_{d = 1}^D \underbrace{\log\left(\KLD{q(X_d | Y)}{p(X_d)} + 1\right)}_{\text{Logarithmic overhead per dimension}}.
    % \Bigg],
    \label{eq:modified_elbo}
\end{align}
\vspace{-6mm}
\pg{Coding the latent indices.} 
As the final step during the encoding process, we need a prefix code to encode the heap indices $I_d$ returned by GRCD for each $d$.
Without any further information, the best we can do is use Elias $\delta$ coding \citep{elias1975universal}, which, assuming our conjecture on the expected runtime of GRCD holds, yields an expected codelength of $\mathbb{I}[Y; X] + 2 \log(\mathbb{I}[Y; X] + 1) + \Oh(1)$.
% As an improvement, we considered an empirical method suggested in \citet{yang2020variational}: we encoded the training set using GRCD and we histogrammed the indices it returned.
% Unfortunately, we found the empirical index distribution too heavy-tailed to be used as the coding distribution.
However, we can improve this if we can estimate $\Exp[\log I_d]$ for each $d$:
it can be shown, that the maximum entropy distribution of a positive integer-valued random variable with under a constraint on the expectation on its logarithm is $\zeta(n | \lambda) \propto n^{-\lambda}$, with $\lambda^{-1} = \Exp[\log I_d] + 1$.
% Instead, we borrow a suggestion from \citet{li2018strong}: we encoded the training set using GRCD, and we estimated $\Exp[\log I_d]$ in each dimension $d$. 
% It can be shown, that the maximum entropy distribution with respect to this constraint is the $\zeta$-distribution, defined over positive integers $n$ by $\zeta(n | \lambda) \propto n^{-\lambda}$, with $\lambda^{-1} = \Exp[\log I_d] + 1$.
In this case, entropy coding $I_d$ using this $\zeta$ distribution yields improves the expected codelength to $\mathbb{I}[Y; X] + \log(\mathbb{I}[Y; X] + 1) + \Oh(1)$.

\vspace{-2mm}
\pg{Experimental results.}
We trained our VAE with $L \in \{20, 50, 100\}$ latent dimensions optimized using the negative ELBO and its modified version in \Cref{eq:modified_elbo}, and experimented with encoding the heap indices of GRCD with both $\delta$ and $\zeta$ coding.
We report the results of our in \Cref{tab:mnist1} on the MNIST test set in bits per pixel.
In addition to the total coding cost, we report the negative ELBO per pixel, which is the fundamental lower bound on the compression efficiency of REC with each VAE.
Finally, we report the logarithmic overhead due to $\delta$ coding.
We find that both the modified ELBO and $\zeta$ coding prove beneficial, especially as the dimensionality of the latent space increases.
This is expected, since the overhead is most significant for latent dimensions with small KLs, which becomes more likely as the dimension of the latent space grows.
The improvements yielded by each of the two methods are significant, with $\zeta$ coding leading to a consistent $1 - 7\%$ gain compared to $\delta$ coding and the modified objective resulting in up to $2\%$ gain in coding performance.
\vspace{-2mm}
\section{Conclusion and Future Work}
\vspace{-2mm}
\pg{Summary.}
In this work, we introduced Greedy Rejection Coding (GRC), a REC algorithm which generalises the rejection algorithm of \citeauthor{harsha2007communication} to arbitrary probability spaces and partitioning processes.
We proved the correctness of our algorithm under mild assumptions, and introduced GRCS and GRCD, two variants of GRC.
We showed that the runtimes of GRCS and GRCD significantly improve upon the runtime of \astar coding, which can be arbitrarily larger.
We evaluated our algorithms empirically, verifying our theory and conducted a proof-of-concept learnt compression experiment on MNIST using VAEs.
We demonstrated that a principled modification to the ELBO and entropy coding GRCD's indices using a $\zeta$ distribution can further improve compression efficiency.

\vspace{-2mm}
\pg{Limitations and Further work.}
One limitation of GRC is that, unlike \astar coding, it requires us to be able to evaluate the CDF of $Q$.
While in some settings this CDF may be intractable, this assumption is satisfied by most latent variable generative models, and is not restrictive in practice.
However, one practical limitation of GRCS and GRCD, as well as \asstar and \adstar, is that they assume target-proposal pairs over $\mathbb{R}$.
For multivariate distributions, we can decompose them into univariate conditionals and apply GRC dimensionwise, however this incurs an additional coding overhead per dimension, resulting in a non-negligible cost.
Thus, an important direction is to investigate whether fast REC algorithms for multivariate distributions can be devised, to circumvent this challenge.

\newpage
\bibliography{references}
\bibliographystyle{neurips_2020_conference}

\newpage
\appendix

% \begin{minipage}[t]{0.49\textwidth}
% \begin{algorithm}
% \caption{GRC with global partitioning process}\label{alg:grcg}
% \begin{algorithmic}
% \Require Target $Q$, proposal $P$, space $\mathcal{X}$ and $D_{\max}$
% \State $d \gets 0, C \gets 1, I_0 = 1$
% \State $T_0 \gets 0, S_1 \gets \mathcal{X}$
% \While{$d \leq D_{\max}$}
%     \State $X_d \sim P|_{S_d}/P(S_d)$
%     \State $U_d \sim \text{Uniform}[0, 1]$
%     \State $\beta_d \gets \texttt{AcceptProb}(X_d, T_0, P, Q)$
%     \State $d \gets d + 1$
% \EndWhile
% \end{algorithmic}
% \end{algorithm}
% \end{minipage}

\section{Formal definition of Greedy Rejection Coding}
\label{app:theory}

\subsection{Formal definition}
\label{app:def:grc}

Here we give a formal definition of GRC in terms of measures.
We chose to omit this from the main text for the sake of exposition, and instead formally define GRC in \cref{def:grc} below.

\begin{definition}[Greedy Rejection Coding]\label{def:grc}
    Let $Z$ be a partitioning process on $\Sigma$, and $I_0 = 1$, $S_0 = Z_{I_0}$.
    Let $T_0(\cdot, S_0)$ be the zero-measure on $(\mathcal{X}, \Sigma)$.
    Then for $d = 0, 1, \dots$ define
    {\normalfont \begin{align}
        t_d(x, S_{0:d}) &\stackrel{\text{def}}{=} \frac{d T_d(\cdot, S_{0:d})}{d P(\cdot)}(x), \label{eq:def:t} \\
        \alpha_{d+1}(x, S_{0:d}) &\stackrel{\text{def}}{=} \min\left\{ \frac{dQ}{dP}(x) - t_d(x, S_{0:d}), \frac{1 - T_d(\mathcal{X}, S_{0:d})}{P(S_d)} \right\} \label{eq:def:alpha} \\
        A_{d+1}(S, S_{0:d}) &\stackrel{\text{def}}{=} \int_S dP(x)~\alpha_{d+1}(x, S_{0:d}), \label{eq:def:A} \\
        \beta_{d+1}(x, S_{0:d}) &\stackrel{\text{def}}{=} \alpha_{d+1}(x, S_{0:d})~\frac{P(S_d)}{1 - T_d(\mathcal{X}, S_{0:d})}, \label{eq:def:beta} \\
        X_{I_d} &\sim \frac{P|_{S_d}}{P(S_d)}, \label{eq:def:X} \\
        U_{I_d} &\sim \text{Uniform}(0, 1), \label{eq:def:U} \\
        b_d &\sim \text{Bernoulli}\left(\frac{Q(Z_{2I_d+1}) - T_d(Z_{2I_d+1}, S_{0:d}) - A_{d+1}(Z_{2I_d+1}, S_{0:d})}{Q(S_d) - T_d(S_d, S_{0:d}) - A_{d+1}(S_d, S_{0:d})}\right), \label{eq:def:b} \\
        I_{d+1} &\stackrel{\text{def}}{=} 2I_d + b_d, \label{eq:def:I} \\
        S_{d+1} &\stackrel{\text{def}}{=} Z_{I_{d+1}}, \label{eq:def:S} \\
        T_{d+1}(S, S_{0:d+1}) &\stackrel{\text{def}}{=} T_d(S \cap S_{d+1}, S_{0:d}) + A_{d+1}(S \cap S_{d+1}, S_{0:d}) + Q(S \cap S_{d+1}'), \label{eq:def:T} 
    \end{align}}where $S \in \Sigma$ and $P|_{Z_d}$ denotes the restriction of the measure $P$ to the set $Z_d$.
    Generalised Greedy Rejection Coding (GRC) amounts to running this recursion, computing
    \begin{align}\label{eq:def:D}
        D^* = \min \{d \in \mathbb{N} : U_{I_d} \leq \beta_{d+1}(X_{I_d}, S_{0:d})\},
    \end{align}
    and returning $X = X_{I_{D^*}}$ and $C = I_{D^*}$.
\end{definition}

The functions \texttt{AcceptProb} and \texttt{RuledOutMass} in \cref{alg:grc} correspond to calculating the quantities in \cref{eq:def:beta} and \cref{eq:def:T}.
The function \texttt{PartitionProb} corresponds to computing the success probability of the Bernoulli coin toss in \cref{eq:def:b}.

\subsection{\texorpdfstring{\citeauthor{harsha2007communication}}{Harsha et al.}'s algorithm is a special case of GRC}
\label{app:harsha:special}

Here we show that the algorithm of \citeauthor{harsha2007communication} is a special case of GRC which assumes discrete $P$ and $Q$ distributions and uses the global partitioning process, which we refer to as GRCG.
Note that the original algorithm described by \citeauthor{harsha2007communication} assumes discrete $P$ and $Q$ distributions, whereas GRCG does not make this assumption.

\begin{proposition}[\cite{harsha2007communication} is a special case of GRC]
\label{app:proposition:grcg}
    Let $Z$ be the global partitioning process over $\Sigma$, defined as
    \begin{equation}
        \label{app:eq:global:Z}
        Z_1 = \mathcal{X}, ~~~ Z_{2n} = Z_n, ~~~Z_{2n+1} = \emptyset, ~~\text{ for all }~ n = 1, 2, \dots.
    \end{equation}
    \cite{harsha2007communication} is equivalent to GRC using this $Z$ and setting $C = D^*$ instead of $C = I_{D^*}$.
    We refer to this variant of GRC as Global GRC, or GRCG for short.
\end{proposition}
\begin{proof}
With $Z$ defined as in \cref{app:eq:global:Z}, we have $b_d \sim \text{Bernoulli}(0)$ by \cref{eq:def:b}, so $b_d = 0$ almost surely.
Therefore $S_d = \mathcal{X}$ for all $d \in \mathbb{N}^+$.
From this, we have $T_{d+1}(S, S_{0:d}) = T_d(S, S_{0:d}) + A_d(S, S_{0:d})$ and also $P(S_d) = P(\mathcal{X}) = 1$ for all $d \in \mathbb{N}^+$.
Substituting these in the equations of \cref{def:grc}, we recover \cref{eq:grcg:1,eq:grcg:2,eq:grcg:3}.
Setting $C = D^*$ instead of $C = I_{D^*}$ makes the two algorithms identical.
\end{proof}

\newpage
\section{Proof of correctness of GRC\texorpdfstring{: \Cref{thm:correctness}}{}}
\label{app:correctness}

In this section we give a proof for the correctness of GRC.
Before going into the proof, we outline our approach and the organisation of the proof.

\pg{Proof outline.}
To prove \cref{thm:correctness}, we consider running GRC for a finite number of $d$ steps.
We consider the measure $\tau_d: \Sigma \to [0, 1]$, defined such that for any $S \in \Sigma$, the quantity $\tau_d(S)$ is equal to the probability that GRC terminates within $d$ steps and returns a sample $X \in S \subseteq \Sigma$.
We then show that $\tau_d \to Q$ in total variation as $d \to \infty$, which proves \cref{thm:correctness}.

\pg{Organisation of the proof.}
First, in section \ref{app:corr:def} we introduce some preliminary definitions, assumptions and notation on partitioning processes, which we will use in later sections.
Then, in \ref{app:tau} we derive the $\tau_d$ measure, and prove some intermediate results about it.
Specifically, \cref{prop:grs:TA} shows that the measures $A_d$ and $T_d$ from the definition of GRC (\cref{def:grc}) correspond to probabilities describing the termination of the algorithm, and \cref{lemma:tau} uses these facts to derive the form of $\tau_d$ in terms of $A_d$.
Then, \cref{lemma:q_minus_tau} shows that the measure $\tau_d$ is no larger than the measure $Q$ and \cref{lem:tau_limit} shows that the limit of $\tau_d$ as $d \to \infty$ is also a measure.
Lastly \cref{lem:T_and_tau} shows that $T_d$ and $\tau_d$ are equal on the active sets of the partition process followed within a run of GRC, and then \cref{lem:Qtau} uses that result to derive the subsets of the sample space on which $\tau_d$ is equal to $Q$ and $\tau$ is equal to $Q$.

Then, in \cref{app:four_cases} we break down the proof of \cref{thm:correctness} in four cases.
First, we consider the probability $p_d$ that GRC terminates at step $d$, given that it has not terminated up to and including step $d - 1$.
\Cref{lemma:pdzero} shows that if $p_d \not \to 0$, then $\tau_d \to Q$ in total variation.
Then we consider the case $p_d \to 0$ and show that in this case, if any of assumptions \ref{assumption:finite_ratio_mode}, \ref{assumption:single_branch} or \ref{assumption:nice_shrinking} hold, then again $\tau_d \to Q$ in total variation.
Putting these results together proves \cref{thm:correctness}.

\subsection{Preliminary definitions, assumptions and notation}
\label{app:corr:def}

For the sake of completeness, we restate relevant definitions and assumptions.
\Cref{def:qp} restates our notation on the target $Q$ and proposal $P$ measures and \cref{as:qp} emphasises our assumption that $Q \ll P$.
\Cref{def:Z} restates the definition of partitioning processes.

% ==========================================
% Proposal and target definition
% ==========================================

{
\begin{definition}[Target $Q$ and proposal $P$ distributions] \label{def:qp}
    Let $Q$ and $P$ be probability measures on a measurable space $(\mathcal{X}, \Sigma)$.
    We refer to $Q$ and $P$ as the target and proposal measures respectively.
\end{definition}
}

% ==========================================
% Assumptions on Q and P
% ==========================================

{
\begin{assumption}[$Q \ll P$] \label{as:qp}
    We assume $Q$ is absolutely continuous w.r.t. $P$, that is $Q \ll P$.
    Under this assumption, the Radon-Nikodym derivative of $Q$ w.r.t. $P$ exists and is denoted as $dQ/dP : \mathcal{X} \to \mathbb{R}^+$.
\end{assumption}
}

% ==========================================
% Partitioning process
% ==========================================

{
\begin{definition}[Partitioning process]
\label{def:Z}
    A random process {\normalfont $Z: \mathbb{N}^+ \to \Sigma$} which satisfies
    {\normalfont \begin{equation}
        Z_1 = \mathcal{X}, ~~ Z_{2n} \cap Z_{2n + 1} = \emptyset, ~~ Z_{2n} \cup Z_{2n + 1} = Z_n.
    \end{equation}}is called a partitioning process.
\end{definition}

That is, a partitioning process $Z$ is a random process indexed by the heap indices of an infinite binary tree, where the root node is $\mathcal{X}$ and any two children nodes $Z_{2n}$ and $Z_{2n + 1}$ partition their parent node $Z_n$.
Note that by definition, a partitioning process takes values which are measurable sets in $(\mathcal{X}, \Sigma)$.
}

% ==========================================
% Definition: Ancestors
% ==========================================

{
Because GRC operates on an binary tree, we find it useful to define some appropriate notation.
\Cref{def:ancestors} specifies the ancestors of a node in a binary tree.
Notation \ref{not:indexing} gives some useful indexing notation for denoting different elements of the partitioning process $Z$, as well as for denoting the branch of ancestors of an element in a partitioning process.

\begin{definition}[Ancestors]
\label{def:ancestors}
We define the one-step ancestor function $A_1: 2^{\mathbb{N}^+} \to 2^{\mathbb{N}^+}$ as
{\normalfont
\begin{align}
    A_1(N) &= N \cup \{n \in \mathbb{N}^+: n' = 2n \text{ or } n' = 2n + 1, \text{ for some } n' \in N\},
\end{align}}
and the ancestor function $A: 2^{\mathbb{N}^+} \to 2^{\mathbb{N}^+}$ as
{\normalfont
\begin{equation}
    A(N) = \left\{n \in \mathbb{N}^+: n \in A_1^k(\{n'\}) \text{ for some } n' \in N, k \in \mathbb{N}^+ \right\}.
\end{equation}}where $A_1^k$ denotes the composition of $A_1$ with itself $k$ times.
\end{definition}
Viewing $\mathbb{N}^+$ as the set of heap indices of an infinite binary tree, $A$ maps a set $N \subseteq \mathbb{N}$ of natural numbers (nodes) to the set of all elements of $N$ and their ancestors.
}

% ==========================================
% Definition: Double indexing
% ==========================================

{
\begin{notation}[Double indexing for $Z$, ancestor branch]
\label{not:indexing}
Given a partitioning process $Z$, we use the notation $Z_{d, k}$, where $d = 1, 2, \dots$ and $k = 1, \dots, 2^{d-1}$ to denote the $k^{th}$ node at depth $d$, that is
\begin{equation}
    Z_{d, k} := Z_{2^{d-1} - 1 + k}.
\end{equation}
We use the hat notation $\hat{Z}_{d, k}$ to denote the sequence of nodes consisting of $Z_{d, k}$ and all its ancestors
\begin{equation}
    \label{eq:ancestor_branch}
    \hat{Z}_{d, k} := (Z_n : n \in A(\{2^{d-1} - 1 + k\})),
\end{equation}
and call $\hat{Z}_{d, k}$ the ancestor branch of $Z_{d, k}$.
\end{notation}

% \begin{notation}[Colon indexing]
% Given a sequence of elements, such as $X_1, \dots, X_N$, we use the following colon notation as a shorthand
% \begin{equation}
%     X_{i:j} := (X_i, \dots, X_j),
% \end{equation}
% with the convention that $X_{i:j}$ is an empty list whenever $i > j$.
% \end{notation}
}

% ==========================================
% Notation: P measure
% ==========================================

{
\begin{notation}[$\mathbb{P}$ measure]
\label{notation:P}
In \cref{def:grc}, we defined $\mathbb{P}$ to be the measure associated with an infinite sequence of independent fair coin tosses over a measurable space $(\Omega, \mathcal{S})$.
To avoid heavy notation, for the rest of the proof we will overload this symbol as follows: if $F$ is a random variable from $\Omega$ to some measurable space, we will abbreviate $\mathbb{P} \circ F^{-1}$ by simply $\mathbb{P}(F)$.
\end{notation}

\subsection{Deriving the measure of samples returned by GRC}
\label{app:tau}
}

% ==========================================
% Acceptance and rejection probabilities
% ==========================================

{
For the remainder of the proof, we condition on a fixed partitioning process sample $Z$.
For brevity, we omit this conditioning which, from here on is understood to be implied.
\Cref{prop:grs:TA} shows that the measures $A_d$ and $T_d$ correspond to the probabilities that GRC picks a particular branch of the binary tree and terminates at step $d$, or does not terminate up to and including step $d$, respectively.

\begin{proposition}[Acceptance and rejection probabilities] \label{prop:grs:TA}
Let $V_d$ be the event that GRC does not terminate up to and including step $d$ and $W_d$ be the event that it terminates at step $d$.
Let $S_{0:d} = B_{0:d}$ denote the event that the sequence of the first $d$ bounds produced is $B_{0:d}$.
Then
\begin{align}
    \mathbb{P}(V_d, S_{0:d} = B_{0:d}) &= 1 - T_d(\mathcal{X}, B_{0:d}), & \text{ for } d = 0, 1, \dots, \label{eq:V} \\
    \mathbb{P}(W_{d+1}, S_{0:d} = B_{0:d}) &= A_{d+1}(\mathcal{X}, B_{0:d}), & \text{ for } d = 0, 1, \dots. \label{eq:W} 
\end{align}
\end{proposition}
\begin{proof}
First we consider the probability that GRC terminates at step $k+1$ given that it has not terminated up to and including step $d$, that is the quantity $\mathbb{P}(W_{k+1} ~|~ V_k, S_{0:k} = B_{0:k})$.
By \cref{def:grc}, this probability is given by integrating the acceptance probability $\beta_{k+1}(x, B_{0:k})$ over $x \in \mathcal{X}$, with respect to the measure $P|_{B_k} / P(B_k)$, that is
\begin{align}
    \mathbb{P}(W_{k+1} ~|~ V_k, S_{0:k} = B_{0:k}) &= \int_{x \in
    B_{k}} dP(x) \frac{\beta_{k+1}(x, B_{0:k})}{P(B_k)} \\
    &= \int_{x \in
    \mathcal{X}} dP(x) \frac{\beta_{k+1}(x, B_{0:k})}{P(B_k)} \\
    &= \int_{x \in \mathcal{X}} dP(x) \frac{\alpha_{k+1}(x, B_{0:k})}{1 - T_k(\mathcal{X}, B_{0:k})} \\
    &= \frac{A_{k+1}(\mathcal{X}, B_{0:k})}{1 - T_k(\mathcal{X}, B_{0:k})}, \label{eq:grs:accept}
\end{align}
Now, we show the result by induction on $d$, starting from the base case of $d = 0$.

\vspace{2mm}
\noindent \textbf{Base case:}
For $d = 0$, by the definition of GRC (\cref{def:grc}) $S_0 = Z_{I_0} = \mathcal{X}$, so
\begin{align}
    \mathbb{P}\left(V_0, S_0 = B_0 \right) = 1 ~\text{ and }~ T_0(\mathcal{X}, B_0) = 0,
\end{align}
which show the base case for \cref{eq:V}.
Now, plugging in $k = 0$ in \cref{eq:grs:accept} we obtain
\begin{equation}
    \mathbb{P}(W_1, S_0 = B_0) = \mathbb{P}(W_1 ~|~ V_0, S_0 = B_0) = \frac{A_1(\mathcal{X}, B_0)}{1 - T_0(\mathcal{X}, B_0)} = A_1(\mathcal{X}, B_0)
\end{equation}
where we have used the fact that $T_0(\mathcal{X}, B_0) = 0$, showing the base case for \cref{eq:W}. \\

\textbf{Inductive step:}
Suppose that for all $k = 0, 1, 2, \dots, d$ it holds that
\begin{equation} \label{eq:grs:inductive}
    \mathbb{P}\left(V_d, S_{0:k} = B_{0:k} \right) = 1 - T_d(\mathcal{X}, B_{0:k}) ~~ \text{ and } ~~ \mathbb{P}\left(W_{k+1}, S_{0:k} = B_{0:k} \right) = A_{k+1}(\mathcal{X}, B_{0:k}).
\end{equation}
Setting $k = d$ in \cref{eq:grs:accept}, we obtain
\begin{align}
    \mathbb{P}(W'_{d+1} ~|~ V_d, S_{0:d} = B_{0:d}) &= \frac{1 - T_d(\mathcal{X}, B_{0:d}) - A_{d+1}(\mathcal{X}, B_{0:d})}{1 - T_d(\mathcal{X}. B_{0:d})},
\end{align}
and using the inductive hypothesis from \cref{eq:grs:inductive}, we have
\begin{align} \label{eq:one_minus_t_minus_a}
    \mathbb{P}(V_{d+1}, S_{0:d} = B_{0:d}) &= \mathbb{P}(W_{d+1}', V_d, S_{0:d} = B_{0:d}) = 1 - T_d(\mathcal{X}, B_{0:d}) - A_{d+1}(\mathcal{X}, B_{0:d}).
\end{align}
Now, $B_d = Z_n$ for some $n \in \mathbb{N}^+$.
Denote $B_d^L := Z_{2n}$ and $B_d^R := Z_{2n+1}$.
Then, by the product rule
\begin{align}
    \mathbb{P}(V_{d+1}, S_{0:d} &= B_{0:d}, S_{d+1} = B^R_d) = \label{eq:ind:0} \\
    \nonumber~\\
    &= \mathbb{P}(S_{d+1} = B^R_d ~|~ V_{d+1}, S_{0:d} = B_{0:d}) \mathbb{P}(V_{d+1}, S_{0:d} = B_{0:d}) \label{eq:ind:1} \\
    \nonumber~\\
    &= \frac{Q(B_d^R) - T_d(B_d^R, B_{0:d}) - A_{d+1}(B_d^R, B_{0:d})}{Q(B_d) - T_d(B_d, B_{0:d}) - A_{d+1}(B_d, B_{0:d})} \mathbb{P}(V_{d+1}, S_{0:d} = B_{0:d}) \label{eq:ind:2} \\
    \nonumber~\\
    &= \frac{Q(B_d^R) - T_d(B_d^R, B_{0:d}) - A_{d+1}(B_d^R, B_{0:d})}{\underbrace{Q(\mathcal{X})}_{=~1} - T_d(\mathcal{X}, B_{0:d}) - A_{d+1}(\mathcal{X}, B_{0:d})} \mathbb{P}(V_{d+1}, B_{0:d} = B_{0:d}) \label{eq:ind:3} \\
    \nonumber~\\
    &= Q(B_d^R) - T_d(B_d^R, B_{0:d}) - A_{d+1}(B_d^R, B_{0:d}) \label{eq:ind:4} \\
    \nonumber~\\
    &= 1 - T_{d+1}(\mathcal{X}, B_{0:d+1}) \label{eq:ind:5}
\end{align}
where we have written $B_{0:d+1} = (B_0, \dots, B_d, B_d^R)$.
Above, to go from \ref{eq:ind:0} to \ref{eq:ind:1} we used the definition of conditional probability, to go from \ref{eq:ind:1} to \ref{eq:ind:2} we used the definition in \ref{eq:def:b}, to go from \ref{eq:ind:2} to \ref{eq:ind:3} we used the fact that for $k = 0, 1, 2, \dots,$ it holds that
\begin{align}
    Q(\mathcal{X}) - T_k(\mathcal{X}, B_{0:k}) - A_{k+1}(\mathcal{X}, B_{0:k}) &= Q(B_k) - T_k(B_k, B_{0:k}) - A_{k+1}(B_k, B_{0:k}) + \nonumber \\
    &~~~~~~+ Q(B_k') - \underbrace{T_k(B_k', B_{0:k})}_{=~Q(B_k')} - \underbrace{A_{k+1}(B_k', B_{0:k})}_{=~0} \\
    &= Q(B_k) - T_d(B_k, B_{0:k}) - A_{k+1}(B_k, B_{0:k}), \label{eq:q_minus_t_minus_a}
\end{align}
from \ref{eq:ind:3} to \ref{eq:ind:4} we have used \cref{eq:one_minus_t_minus_a}, and lastly from \ref{eq:ind:4} to \ref{eq:ind:5} we have again used \cref{eq:q_minus_t_minus_a}.
\Cref{eq:ind:5} similarly holds if $B_{d+1} = B^R_d$ by $B_{d+1} = B^L_d$, so we arrive at
\begin{align} \label{eq:indt}
    \mathbb{P}(V_{d+1}, B_{0:d+1} &= B_{0:d+1}) = 1 - T_{d+1}(\mathcal{X}, B_{0:d+1}),
\end{align}
which shows the inductive step for \cref{eq:V}.
Further, we have
\begin{align}
     \mathbb{P}(W_{d+2}, B_{0:d+1} = B_{0:d+1}) = \mathbb{P}(W_{d+2} ~|~ V_{d+1}, B_{0:d+1} = B_{0:d+1}) \mathbb{P}(V_{d+1}, B_{0:d+1} = B_{0:d+1})
\end{align}
and also by setting $k = d + 1$ in \cref{eq:grs:accept} we have
\begin{equation} \label{eq:indw}
    \mathbb{P}(W_{d+2} ~|~ V_{d+1}, B_{0:d+1} = B_{0:d+1}) = \frac{A_{d+2}(\mathcal{X}, B_{0:d+1})}{1 - T_{d+1}(\mathcal{X}, B_{0:d+1})}.
\end{equation}
Combining \cref{eq:indt} and \cref{eq:indw} we arrive at
\begin{align} \label{eq:inda}
     \mathbb{P}(W_{d+2}, B_{0:d+1} = B_{0:d+1}) = A_{d+2}(\mathcal{X}, B_{0:d+1}),
\end{align}
which is the inductive step for \cref{eq:W}.
Putting \cref{eq:indt,eq:inda} together shows the result.
\end{proof}
}

% ==========================================
% Density of generated samples
% ==========================================

{
We now turn to defining and deriving the form of the measure $\tau_D$.
We will define $\tau_D$ to be the measure such that for any $S \in \Sigma$, the probability that GRC terminates up to and including step $D$ and returns a sample within $S$ is given by $\tau_D(S)$.
We will also show that $\tau_D$ is non-increasing in $D$.

\begin{lemma}[Density of samples generated by GRC] \label{lemma:tau}
The probability that GRC terminates by step $D \geq 1$ and produces a sample in $S$ is given by the measure
\begin{equation} \label{eq:lemma:tau}
    \tau_D(S) = \sum_{d=1}^D \sum_{k = 1}^{2^{d-1}} A_d(S, \hat{Z}_{d, k}),
\end{equation}
where $\hat{Z}_{D, k}$ is the ancestor branch of $Z_{D, k}$ as defined in \cref{eq:ancestor_branch}.
Further, $\tau_D$ is non-decreasing in $D$, that is if $n \leq m$, then $\tau_n(S) \leq \tau_m(S)$ for all $S\in \Sigma$.
\end{lemma}
\begin{proof}
Let $V_d$ be the event that GRC does not terminate up to and including step $d$ and let $W_d(S)$ be the event that GRC terminates at step $d$ and returns a sample in $S$.
Then
\begin{align}
    \tau_D(S) &= \sum_{d=1}^D \mathbb{P}(W_d(S)) \\
              &= \sum_{d=1}^D \mathbb{P}(W_d(S), V_{d-1}) \\
              &= \sum_{d=1}^D \sum_{k=1}^{2^{d-1}} \mathbb{P}(W_d(S), V_{d-1}, S_{0:d-1} = \hat{Z}_{d, k}) \\
              &= \sum_{d=1}^D \sum_{k=1}^{2^{d-1}} \mathbb{P}(W_d(S) ~|~ V_{d-1}, S_{0:d-1} = \hat{Z}_{d, k})~\mathbb{P}(V_{d-1}, S_{0:d-1} = \hat{Z}_{d, k}).
              \label{eq:sumwv}
\end{align}
Further, the terms in the summand can be expressed as
\begin{align}
\mathbb{P}(V_{d-1}, S_{0:d-1} = \hat{Z}_{d, k}) &= 1 - T_{d-1}(\mathcal{X}, \hat{Z}_{d, k}), \label{eq:v_equal_1_minus_T} \\
\mathbb{P}(W_d(S) ~|~ V_{d-1}, S_{0:d-1} = \hat{Z}_{d, k}) &= \int_{x \in S} dP(x) \frac{\beta_{d}(x, \hat{Z}_{d, k})}{P(Z_{d, k})} \\
&= \int_{x \in S} dP(x) \frac{\alpha_d(x, \hat{Z}_{d, k})}{1 - T_{d-1}(\mathcal{X}, \hat{Z}_{d, k})} \\
&= \frac{A_d(S, \hat{Z}_{d, k})}{1 - T_{d-1}(\mathcal{X}, \hat{Z}_{d, k})},
\label{eq:ws}
\end{align}
and substituting \cref{eq:v_equal_1_minus_T,eq:ws} into the sum in \cref{eq:sumwv}, we obtain \cref{eq:lemma:tau}.
Further, since the inner summand is always non-negative, increasing $D$ adds more non-negative terms to the sum, so $\tau_D$ is also non-decreasing in $D$.
\end{proof}
}

% ==========================================
% Q minus tau is non-negative
% ==========================================

{
Now we turn to proving a few results about the measure $\tau_D$.
\Cref{lemma:q_minus_tau} shows that $\tau_D \leq Q$ for all $D$.
This result implies that $||Q - \tau_D||_{TV} = Q(\mathcal{X}) - \tau_D(\mathcal{X})$, which we will use later.

\begin{lemma}[$Q - \tau_D$ is non-negative] \label{lemma:q_minus_tau}
Let $D \in \mathbb{N}^+$.
Then $Q - \tau_D$ is a positive measure, that is
\begin{equation}
    Q(S) - \tau_D(S) \geq 0 \text{ for any } S \in \Sigma.
\end{equation}
\end{lemma}
\begin{proof}
Let $S \in \Sigma$ and write
\begin{align}
    Q(S) - \tau_D(S) &= \sum_{k=1}^{2^{D-1}} Q(S \cap Z_{D, k}) - \tau_D(S \cap Z_{D, k}) \\
    &= \sum_{k=1}^{2^{D-1}} \left[ Q(S \cap Z_{D, k}) - \sum_{d=1}^D \sum_{k' = 1}^{2^{D-1}} A_d(S \cap Z_{D, k}, \hat{Z}_{D, k'}) \right] \\
    &= \sum_{k=1}^{2^{D-1}} \left[ Q(S \cap Z_{D, k}) - \sum_{d=1}^D A_d(S \cap Z_{D, k}, \hat{Z}_{D, k}) \right] \\
    &= \sum_{k=1}^{2^{D-1}} \left[ Q(S \cap Z_{D, k}) - T_{D-1}(S \cap Z_{D, k}, \hat{Z}_{D, k}) - A_D(S \cap Z_{D, k}, \hat{Z}_{D, k}) \right] \label{eq:q_minus_t_positive}
\end{align}
We will show that the summand in \cref{eq:q_minus_t_positive} is non-negative.
From the definition in \cref{eq:def:alpha} we have
\begin{align}
    \alpha_D(x, \hat{Z}_{D, k}) &= \min\left\{ \frac{dQ}{dP}(x) - t_{D-1}(x, \hat{Z}_{D, k}), \frac{1 - T_{D-1}(\mathcal{X}, \hat{Z}_{D, k})}{P({Z_{D, k}})} \right\} \\
    &\leq \frac{dQ}{dP}(x) - t_{D-1}(x, \hat{Z}_{D, k}) \label{eq:dqdp_minus_t}
\end{align}
and integrating both sides of \cref{eq:dqdp_minus_t} over $S \cap Z_{D, k}$, we obtain
\begin{align}
& A_D(S \cap Z_{D, k}, \hat{Z}_{D, k}) \leq Q(S \cap Z_{D, k}) - T_{D-1}(S \cap Z_{D, k}, \hat{Z}_{D, k})
\end{align}
Putting this together with \cref{eq:q_minus_t_positive} we arrive at
\begin{align}
    Q(S) - \tau_D(S) \geq 0,
\end{align}
which is the required result.
\end{proof}
}

% ==========================================
% Lemma: tau is a measure
% ==========================================

{
Thus far we have derived the form of $\tau_D$, shown that it is non-decreasing in $D$ and that it is no greater than $Q$.
As we are interested in the limiting behaviour of $\tau_D$, we next show that its limit, $\tau = \lim_{D \to \infty} \tau_D$, is also a measure.
Further, it also holds that $\tau \leq Q$.

\begin{lemma}[Measures $\tau_D$ converge to a measure $\tau \leq Q$]
    \label{lem:tau_limit}
    For each $S \in \Sigma$, $\tau_D(S)$ converges to a limit.
    Further, the function $\tau : \Sigma \to [0, 1]$ defined as
    \begin{equation}
        \tau(S) = \lim_{D\to\infty} \tau_D(S)
    \end{equation}
    is a measure on $(\mathcal{X}, \Sigma)$ and $\tau(S) \leq Q(S)$ for all $S \in \Sigma$.
\end{lemma}
\begin{proof}
    First, by \cref{lemma:tau}, $\tau_D(S)$ is non-decreasing in $D$, and bounded above by $Q(S)$ for all $S \in \Sigma$.
    Therefore, for each $S \in \Sigma$, $\tau_D(S)$ converges to some limit as $D \to \infty$.
    Define $\tau : \Sigma \to [0, 1]$ as
    \begin{equation}
        \tau(S) = \lim_{D\to\infty} \tau_D(S),
    \end{equation}
    and note that $\tau$ is a non-negative set function for which $\tau(\emptyset) = 0$.
    By the Vitali-Hahn-Saks theorem \cite[see Corollary 4, p. 160;][]{dunford1988linear}, $\tau$ is also countably additive, so it is a measure.
    Also, by \cref{lemma:q_minus_tau}, $\tau_D(S) \leq Q(S)$ for all $D \in \mathbb{N}^+$ and all $S \in \Sigma$, so $\tau(S) \leq Q(S)$ for all $S \in \Sigma$.
\end{proof}
}

% ==========================================
% Definition: H sets
% ==========================================

\begin{definition}[$H_{d, k}$, $H_d$ and $H$]
For $d = 1, 2, \dots$ and $k = 1, \dots, 2^{d-1}$, we define the sets $H_{d, k}$ as
\begin{equation}
    \label{eq:def_hdk}
    H_{d, k} = \left\{x \in Z_{d, k} ~\Big|~ \frac{dQ}{dP}(x) - t_{d-1}(x, \hat{Z}_{d, k}) \geq \frac{1 - T_{d-1}(\mathcal{X}, \hat{Z}_{d, k})}{P(Z_{d, k})} \right\}.
\end{equation}
Also, define the sets $H_d$ and $H$ as
\begin{align}
    H_d = \bigcup_{k = 1}^{2^{d-1}} H_{d, k} ~ \text{ and } ~ H = \bigcap_{d = 1}^{\infty} H_d.
\end{align}
\end{definition}

% ==========================================
% Lemma: T_d = \tau_d on Z
% ==========================================

\begin{lemma}[$T_D(\cdot, \hat{Z}_{D+1, k})$ and $\tau_D$ agree in $Z_{D+1, k}$]
\label{lem:T_and_tau}
Let $R \in \Sigma$.
If $R \subseteq Z_{D+1, k}$, then
\begin{equation}
     \tau_D(R) = T_D(R, \hat{Z}_{D+1, k}).
\end{equation}
\end{lemma}
\begin{proof}
    Suppose $R \subseteq Z_{D+1, k}$.
    First, we have
    \begin{equation}
        \tau_D(R) = \sum_{d=1}^D \sum_{k'=1}^{2^{d-1}} A_d(R, \hat{Z}_{d, k'}) = \sum_{d=1}^D A_d(R, (\hat{Z}_{D+1, k})_{1:d}).
    \end{equation}
    From the definition of $T_D$ in \cref{eq:def:T}, we have
    \begin{align}
        T_D(R, \hat{Z}_{D+1, k}) &= T_{D-1}(R \cap Z_{D+1, k}, (\hat{Z}_{D+1, k})_{1:D}) + A_D(R \cap Z_{D+1, k}, (\hat{Z}_{D+1, k})_{1:D}) + \label{eq:recursive1} \\
        &\quad \quad + \underbrace{Q(R \cap Z_{D+1, k}')}_{=~0} \nonumber \\
        &= T_{D-1}(R \cap Z_{D+1, k}, (\hat{Z}_{D+1, k})_{1:D}) + A_D(R \cap Z_{D+1, k}, (\hat{Z}_{D+1, k})_{1:D}) \label{eq:recursive2} \\
        &= T_{D-1}(R, (\hat{Z}_{D+1, k})_{1:D}) + A_D(R, (\hat{Z}_{D+1, k})_{1:D}) \label{eq:recursive3}
    \end{align}
    where we have used the assumption that $R \subseteq Z_{D+1, k}$.
    In a similar manner, applying \cref{eq:recursive3} recursively $D - 1$ more times, we obtain
    \begin{equation}
        T_D(R, \hat{Z}_{D+1, k}) =  \sum_{d=1}^D A_d(R, (\hat{Z}_{D+1, k})_{1:d}) = \tau_D(R).
    \end{equation}
    which is the required result.
\end{proof}

% ==========================================
% Lemma: tau(Omega \ H) = Q(Omega \ H)
% ==========================================

\begin{lemma}[Equalities with $Q, \tau_D$ and $\tau$]
\label{lem:Qtau}
The following two equalities hold
\begin{equation}
    Q(\mathcal{X} \setminus H_D) = \tau_D(\mathcal{X} \setminus H_D) ~\text{ and }~ Q(\mathcal{X} \setminus H) = \tau(\mathcal{X} \setminus H).
\end{equation}
\end{lemma}
\begin{proof}
    Let $R = Z_{D+1, k} \setminus H_{D, k}$. 
    Then, by similar reasoning used to prove \cref{eq:recursive1}, we have
    \begin{align}
        \label{eq:TTQ}
        T_D(R, \hat{Z}_{D+1,k}) = T_{D-1}(R, (\hat{Z}_{D+1, k})_{1:D}) + A_D(R, (\hat{Z}_{D+1, k})_{1:D})
    \end{align}
    Further, we also have
    \begin{align}
        A_D(R, \hat{Z}_{D, k}) &= \int_R dP(x)~\alpha_D(x, \hat{Z}_{D, k}) \label{eq:AQT1} \\
        &= \int_R dP(x)~ \min\left\{ \frac{dQ}{dP}(x) - t_{D-1}(x, \hat{Z}_{D, k}), \frac{1 - T_{D-1}(\mathcal{X}, \hat{Z}_{D, k})}{P(Z_{D, k})} \right\} \label{eq:AQT2} \\
        &= \int_R dP(x)~ \left(\frac{dQ}{dP}(x) - t_{D-1}(x, \hat{Z}_{D, k}) \right) \label{eq:AQT3} \\
        &= Q(R) - T_{D-1}(R, \hat{Z}_{D, k})
        \label{eq:AQT4}
    \end{align}
    where from \cref{eq:AQT2} to \cref{eq:AQT3} we have used the definition of $H_{D, k}$.
    Then, combining \cref{eq:TTQ,eq:AQT4} and using \cref{lem:T_and_tau}, we arrive at
    \begin{equation}
        \label{eq:R}
         Q(Z_{D+1, k} \setminus H_{D, k}) = T_D(Z_{D+1, k} \setminus H_{D, k}, \hat{Z}_{D+1,k}) = \tau_D(Z_{D+1, k} \setminus H_{D, k}).
    \end{equation}
    Now, using the equation above, we have that
    \begin{equation}
        \tau_D(\mathcal{X} \setminus H_D) = \sum_{k=1}^{2^D} \tau_D(Z_{D+1,k} \setminus H_D) = \sum_{k=1}^{2^D} Q(Z_{D+1,k} \setminus H_D) = Q(\mathcal{X} \setminus H_D).
    \end{equation}
    Now, using $\tau_D \leq \tau \leq Q$ and $\tau_D(\mathcal{X} \setminus H_D) = Q(\mathcal{X} \setminus H_D)$, we have that $\tau(\mathcal{X} \setminus H_D) = Q(\mathcal{X} \setminus H_D)$, which is the first part of the result we wanted to show.
    Taking limits, we obtain
    \begin{equation}
        Q(\mathcal{X} \setminus H) = \lim_{D \to \infty} Q(\mathcal{X} \setminus H_D) = \lim_{D \to \infty} \tau(\mathcal{X} \setminus H_D) = \tau(\mathcal{X} \setminus H),
    \end{equation}
    which is the second part of the required result.
\end{proof}

\newpage
\subsection{Breaking down the proof of \Cref{thm:correctness} in five cases}
\label{app:four_cases}

In \cref{def:w} we introduce the quantities $w_d = Q(\mathcal{X}) - \tau_d(\mathcal{X})$ and $p_d = \mathbb{P}(W_d~|~V_{d-1})$.
Then we break down the proof of \cref{thm:correctness} in five cases.
First, in \cref{lemma:pdzero} we show that if $p_d \not \to 0$, then $w_d \to 0$.
Second, in \cref{lemma:hzero} we show that if $P(H_d) \to 0$, then $w_d \to 0$.
In \cref{lemma:intermediate} we show an intermediate result, used in the other three cases, which we consider in lemmas \ref{lemma:finite_ratio}, \ref{lemma:single_branch} and \ref{lemma:nicely_shrkinking}.
Specifically, in these three cases we show that if $p_d \to 0$ and $P(H_d) \not \to 0$, and assumption \ref{assumption:finite_ratio_mode}, \ref{assumption:single_branch} or \ref{assumption:nice_shrinking} hold respectively, we have $w_d \to 0$.
Putting these results together shows \cref{thm:correctness}.

% ==========================================
% Definition: w_d and p_d
% ==========================================

{
\begin{definition}[$p_d$, $w_{d, k}$ and $w_d$]
\label{def:w}
Define $p_d = \mathbb{P}(W_d~|~V_{d-1})$.
Also define $w_{d, k}$ and $w_d$ as
{
\normalfont
\begin{align}
    w_{d,k} &\stackrel{\text{def}}{=} Q(Z_{d, k}) - \tau_d(Z_{d, k}), \\
    w_d &\stackrel{\text{def}}{=} \sum_{k=1}^{2^{d-1}} w_{d,k}.
\end{align}
}
\end{definition}
}

\begin{lemma}[$w_d$ non-increasing in $d$]
    The sequence $w_d$ is non-negative and non-increasing in $d$.
\end{lemma}
\begin{proof}
Since $\tau_d$ is non-decreasing in $d$ (from \cref{lem:Qtau}) and
\begin{equation}
    w_d = \sum_{k=1}^{2^{d-1}} Q(Z_{d, k}) - \tau_d(Z_{d, k}) = Q(\mathcal{X}) - \tau_d(\mathcal{X}),
\end{equation}
it follows that $w_d$ is a non-increasing and non-negative sequence.
\end{proof}

% ==========================================
% Lemma: first case
% ==========================================

\begin{lemma}[Case 1]
\label{lemma:pdzero}
    If $p_d \not \to 0$, then $w_d \to 0$.
\end{lemma}
\begin{proof}
Let $p_d = \mathbb{P}(W_d~|~V_{d-1})$ and suppose $p_d \not \to 0$.
Then, there exists $\epsilon > 0$ such that $p_d > \epsilon$ occurs infinitely often.
Therefore, there exists an increasing sequence of integers $a_d \in \mathbb{N}$ such that $p_{a_d} > \epsilon$ for all $d \in \mathbb{N}$.
Then
\begin{align}
    \tau_{a_d}(\mathcal{X}) &= \mathbb{P}\left(\bigcup_{d=1}^{a_d} W_d\right) \\
    &= 1 - \mathbb{P}\left(V_{a_d}\right), \\
    &= 1 - \prod_{d=1}^{a_d} \mathbb{P}\left(V_d~|~V_{d-1}\right), \\
    &= 1 - \prod_{d=1}^{a_d} (1 - p_d), \\
    &\geq 1 - (1 - \epsilon)^d \to 1 \text{ as } d \to \infty.
    \label{eq:tau_to_1}
\end{align}
Therefore, $\tau_d(\mathcal{X}) \to 1$ as $d \to \infty$, which implies that $||Q - \tau_d||_{TV} \to 0$.
\end{proof}

% ==========================================
% Lemma: second case
% ==========================================

\begin{lemma}[Case 2]
\label{lemma:hzero}
    If $P(H_d) \to 0$, then $w_d \to 0$.
\end{lemma}
\begin{proof}
    Suppose $P(H_d) \to 0$.
    Since $Q \ll P$, we have $Q(H) = 0$, and since $Q \geq \tau \geq 0$ (by \cref{lem:tau_limit}), we also have $\tau(H) = 0$.
    Therefore
    \begin{align}
        \lim_{d \to \infty} w_d &= \lim_{d \to \infty} ||Q - \tau_d||_{TV} \\
        &= Q(\mathcal{X}) - \tau(\mathcal{X}) \\
        &= \underbrace{Q(\mathcal{X} \setminus H) - \tau(\mathcal{X} \setminus H)}_{=~0 \text{ from lemma } \ref{lem:Qtau}} + \underbrace{Q(H)}_{=~0} - \underbrace{\tau(H)}_{=~0} \\
        &= 0
    \end{align}
    which is the required result.
\end{proof}

% ==========================================
% Lemma: intermediate lemma
% ==========================================

\begin{lemma}[An intermediate result]
\label{lemma:intermediate}
    If $p_d \to 0$ and $w_d \not \to 0$ as $d \to \infty$, then
    \begin{equation}
        \sum_{k=1}^{2^{d-1}} \frac{P(H_{d, k})}{P(Z_{d, k})} ~ w_{d,k} \to 0 \text{ as } d \to \infty.
    \end{equation}
\end{lemma}
\begin{proof}
Suppose that $p_d = \mathbb{P}(W_d~|~V_{d-1}) \to 0$ and $w_d \not \to 0$.
Then
\begin{align}
    \mathbb{P}(W_d~|~V_{d-1}) &\geq \mathbb{P}(W_d(H_d) ~|~ V_{d-1}) \\
        &= \sum_{k=1}^{2^{d-1}} \mathbb{P}\left(W_d(H_{d, k}) ~|~ V_{d-1}\right) \\
        &= \sum_{k=1}^{2^{d-1}} \mathbb{P}\left(W_d(H_{d, k}), S_{0:d-1} = \hat{Z}_{d,k} ~|~ V_{d-1}\right) \\
        &= \sum_{k=1}^{2^{d-1}} \mathbb{P}\left(W_d(H_{d, k})~|~ V_{d-1}, S_{0:d-1} = \hat{Z}_{d,k}\right) \mathbb{P}\left(S_{0:d-1} = \hat{Z}_{d,k} ~|~ V_{d-1}\right) \\
        &= \sum_{k=1}^{2^{d-1}} \frac{P(H_{d, k})}{P(Z_{d, k})} \mathbb{P}\left(S_{0:d-1} = \hat{Z}_{d,k} ~|~ V_{d-1}\right) \\
        &= \sum_{k=1}^{2^{d-1}} \frac{P(H_{d, k})}{P(Z_{d, k})} \frac{w_{d,k}}{w_d} \to 0.
\end{align}
In addition, if $w_d \not \to 0$, then since $0 \leq w_d \leq 1$ we have
\begin{equation}
    \sum_{k=1}^{2^{d-1}} \frac{P(H_{d, k})}{P(Z_{d, k})} w_{d, k} \to 0.
\end{equation}
which is the required result.
\end{proof}

% ==========================================
% Lemma: second case
% ==========================================

\begin{lemma}[Case 3]
    \label{lemma:finite_ratio}
    Suppose that $p_d \to 0$, $P(H_d) \not \to 0$ and \cref{assumption:finite_ratio_mode} holds.
    Then $w_d \to 0$.
\end{lemma}
\begin{proof}
Suppose that $p_d \to 0$, $P(H_d) \not \to 0$.
Suppose also that \cref{assumption:finite_ratio_mode} holds, meaning there exists $M \in \mathbb{R}$ such that $dQ/dP(x) < M$ for all $x \in \mathcal{X}$.
Then for any $S \in \Sigma$, we have
\begin{equation}
    \label{eq:upper_bound}
    \frac{Q(S) - \tau(S)}{P(S)} \leq \frac{Q(S)}{P(S)} =  \frac{\int_S \frac{dQ}{dP} dP}{P(S)} \leq M~ \frac{\int_S dP}{P(S)} = M \implies \frac{Q(S) - \tau(S)}{M} \leq P(S).
\end{equation}
Further, we have
\begin{align}
    \sum_{k=1}^{2^{d-1}} \frac{P(H_{d, k})}{P(Z_{d, k})} ~w_{d,k} &\geq \sum_{k=1}^{2^{d-1}} \frac{P(H_{d, k})}{P(Z_{d, k})} ~ (Q(H_{d, k}) - \tau(H_{d, k})) \\
    &\geq \frac{1}{M} \sum_{k=1}^{2^{d-1}} \frac{(Q(H_{d, k}) - \tau(H_{d, k}))^2}{P(Z_{d, k})} \\
    &\geq \frac{1}{M} \sum_{k=1}^{2^{d-1}} \frac{(Q(H \cap H_{d, k}) - \tau(H \cap H_{d, k}))^2}{P(Z_{d, k})} \\
    &\geq \frac{1}{M} \sum_{k=1}^{2^{d-1}} \frac{\Delta_{d, k}^2}{P(Z_{d, k})} \\
    &= \frac{1}{M} ~ \Phi_d \\
    &\to 0,
\end{align}
where in the second inequality we have used \cref{eq:upper_bound} and we have defined
\begin{align}
    \Delta_{d, k} &\stackrel{\text{def}}{=} Q(H \cap H_{d, k}) - \tau(H \cap H_{d, k}), \\
    \Phi_d &\stackrel{\text{def}}{=} \sum_{k=1}^{2^{d-1}} \frac{\Delta_{d, k}^2}{P(Z_{d, k})}. \label{eq:def_phi}
\end{align}
Now note that the sets $H \cap H_{d + 1, 2k}$ and $H \cap H_{d + 1, 2k+1}$ partition the set $H \cap H_{d, k}$.
Therefore
\begin{equation}
    \Delta_{d, k} = \Delta_{d+1, 2k} + \Delta_{d+1, 2k+1}. \label{eq:deltas_sum}
\end{equation}
By the definition of $\Phi_d$ in \cref{eq:def_phi}, we can write
\begin{equation}
    \Phi_{d + 1} = \sum_{k=1}^{2^d} \frac{\Delta_{d, k}^2}{P(Z_{d+1, k})} = \sum_{k=1}^{2^{d-1}}\left[ \frac{\Delta_{d+1, 2k}^2}{P(Z_{d+1, 2k})} + \frac{\Delta_{d+1, 2k+1}^2}{P(Z_{d+1, 2k+1})}  \right],
\end{equation}
where we have written the sum over $2^d$ terms as a sum over $2^{d-1}$ pairs of terms.
We can rewrite the summand on the right hand side as
\begin{align}
    \frac{\Delta_{d+1, 2k}^2}{P(Z_{d+1, 2k})} + \frac{\Delta_{d+1, 2k+1}^2}{P(Z_{d+1, 2k+1})} &= \frac{\Delta_{d+1, 2k}^2}{P(Z_{d+1, 2k})} + \frac{(\Delta_{d, k} - \Delta_{d+1, 2k})^2}{P(Z_{d+1, 2k+1})} \label{eq:delta2_sum1} \\
    &= \Delta_{d, k}^2 \left[\frac{\rho^2}{P(Z_{d+1, 2k-1})} + \frac{(1 - \rho)^2 }{P(Z_{d+1, 2k})}\right] \label{eq:delta2_sum2} \\
    &= \Delta_{d, k}^2 ~ g(\rho) \label{eq:delta2_sum3}
\end{align}
where in \cref{eq:delta2_sum1} we have used \cref{eq:deltas_sum}, from \cref{eq:delta2_sum1} to \cref{eq:delta2_sum2} we defined the quantity $\rho = \Delta_{d+1, 2k} / \Delta_{d, k}$, and from \cref{eq:delta2_sum2} to \cref{eq:delta2_sum3} we have defined $g : [0, 1] \to \mathbb{R}$ as
\begin{equation}
    g(r) \stackrel{\text{def}}{=} \frac{r^2}{P(Z_{d+1, 2k})} + \frac{(1 - r)^2}{P(Z_{d+1, 2k+1})}.
\end{equation}
The first and second derivatives of $g$ are
\begin{align}
    \frac{dg}{dr} &= \frac{2r}{P(Z_{d+1, 2k})} - \frac{2(1 - r)}{P(Z_{d+1, 2k+1})}, \\
    \frac{d^2g}{dr^2} &= \frac{2}{P(Z_{d+1, 2k})} + \frac{2}{P(Z_{d+1, 2k+1})} > 0,
\end{align}
so $g$ has a single stationary point that is a minimum, at $r = r_{\min}$, which is given by
\begin{align}
    r_{\min} := \frac{P(Z_{d+1, 2k})}{P(Z_{d+1, 2k}) + P(Z_{d+1, 2k+1})}.
\end{align}
Plugging this back in $g$, we obtain
\begin{align}
    g(r_{\min}) = \frac{1}{P(Z_{d+1, 2k}) + P(Z_{d+1, 2k+1})} = \frac{1}{P(Z_{d, k})},
\end{align}
which implies that
\begin{align}
\frac{\Delta_{d+1, 2k}^2}{P(Z_{d+1, 2k})} + \frac{\Delta_{d+1, 2k+1}^2}{P(Z_{d+1, 2k+1})} \geq \frac{\Delta_{d, k}^2}{P(Z_{d, k})}.
\end{align}
Therefore
\begin{equation}
    \Phi_{d + 1} = \sum_{k=1}^{2^d} \frac{\Delta_{d, k}^2}{P(Z_{d+1, k})} \geq \sum_{k=1}^{2^{d-1}} \frac{\Delta_{d, k}^2}{P(Z_{d, k})} = \Phi_d,
\end{equation}
but since $\Phi_d \to 0$, this is only possible if $\Phi_d = 0$ for all $d$, including $d = 1$, which would imply that
\begin{equation}
    \Delta_{1, 1} = Q(H \cap H_{1, 1}) - \tau(H \cap H_{1, 1}) = Q(H) - \tau(H) = 0, \label{eq:p1_qt0}
\end{equation}
which, together with \cref{lem:Qtau}, implies that
\begin{equation}
    Q(\mathcal{X}) - \tau(\mathcal{X}) = Q(H) - \tau(H) = 0,
\end{equation}
and therefore $w_d = ||Q - \tau_d||_{TV} \to 0$.
\end{proof}

% ==========================================
% Lemma: fourth case
% ==========================================

\begin{lemma}[Case 4]
    \label{lemma:single_branch}
    Suppose that $p_d \to 0$, $P(H_d) \not \to 0$ and \cref{assumption:nice_shrinking} holds.
    Then $w_d \to 0$.
\end{lemma}
\begin{proof}
    Suppose that $p_d \to 0$, $P(H_d) \not \to 0$.
    Suppose also that assumption that \cref{assumption:nice_shrinking} holds, meaning that for each $d$, we have $w_{d, k} > 0$ for exactly one value of $k = k_d$, and $w_{d, k} = 0$ for all other $k \neq k_d$.
    In this case, it holds that $H_{d, k} = \emptyset$ for all $k \neq k_d$ and $H_d = H_{d, k_d}$.
    Since $P(H_d) \not \to 0$ and $P(H_d)$ is a decreasing sequence, it converges to some positive constant.
    We also have
    \begin{equation}
        p_d \geq \sum_{k=1}^{2^{d-1}} \frac{P(H_{d, k})}{P(Z_{d, k})} ~w_{d,k} = \frac{P(H_{d, k_d})}{P(Z_{d, k_d})} ~w_{d,k_d} = \frac{P(H_{d, k_d})}{P(Z_{d, k_d})} ~w_d \geq P(H_d)~w_d \to 0,
    \end{equation}
    which can only hold if $w_d \to 0$, arriving at the result.
\end{proof}

\begin{lemma}[Case 5]
    \label{lemma:nicely_shrkinking}
    Suppose that $p_d \to 0$, $P(H_d) \not \to 0$ and \cref{assumption:nice_shrinking} holds.
    Then $w_d \to 0$.
\end{lemma}
\begin{proof}
Suppose that $p_d \to 0$, $P(H_d) \not \to 0$ and \cref{assumption:nice_shrinking} holds.
Since each $x \in \mathcal{X}$ belongs to exactly one $Z_{d, k}$ we can define the function $B_d : \mathcal{X} \to \Sigma$ as
\begin{equation}
    \label{eq:defBd}
    B_d(x) = Z_{d, k} \text{ such that } x \in Z_{d, k}.
\end{equation}
Using this function we can write
\begin{equation*}
    p_d \geq \sum_{k=1}^{2^{d-1}} \frac{P(H_{d, k})}{P(Z_{d, k})} ~w_{d,k} = \sum_{k=1}^{2^{d-1}} P(H_{d, k}) ~ \frac{Q(Z_{d, k}) - \tau_d(Z_{d, k})}{P(Z_{d, k})} = \int_{H_d} dP ~ \frac{Q(B_d(x)) - \tau_d(B_d(x))}{P(B_d(x))}.
\end{equation*}
Now, because the sets $H_d$ are measurable, their intersection $H := \cap_{d=1}^\infty H_d$ is also measurable.
We can therefore lower bound the integral above as follows
\begin{align}
    \label{eq:integral_form}
    \int_{H_d} dP ~ \frac{Q(B_d(x)) - \tau_d(B_d(x))}{P(B_d(x))} &\geq \int_H dP ~ \frac{Q(B_d(x)) - \tau_d(B_d(x))}{P(B_d(x))} \\
    &\geq \int_H dP ~ \frac{Q(B_d(x)) - \tau(B_d(x))}{P(B_d(x))},
\end{align}
where the first inequality holds as the integrand is non-negative and we are constraining the integration domain to $H \subseteq H_d$, and the second inequality holds because $\tau_d(S) \leq \tau(S)$ for any $S \in \Sigma$.
Define $\mathcal{C}$ to be the set of all intersections of nested partitions, with non-zero mass under $P$
\begin{equation}
    \mathcal{C} = \left\{\bigcap_{d = 0}^\infty Z_{d,k_d} : P\left(\bigcap_{d = 0}^\infty Z_{d,k_d}\right) > 0, k_0 = 1, k_{d+1} = 2k_d \text{ or } k_{d+1} = 2k_d + 1 \right\},
\end{equation}
and note that all of its elements are pairwise disjoint.
Each of the elements of $\mathcal{C}$ is a measurable set because it is a countable intersection of measurable sets.
In addition, $\mathcal{C}$ is a countable set, which can be shown as follows.
Define the sets $\mathcal{C}_n$ as
\begin{equation}
    \mathcal{C}_n = \left\{E \in \mathcal{C} : 2^{-n - 1} < P(E) \leq 2^{-n} \right\} \text{ for } n = 0, 1, \dots
\end{equation}
and note that their union equals $\mathcal{C}$.
Further, note that each $\mathcal{C}_n$ must contain a finite number of elements.
That is because if $\mathcal{C}_n$ contained an infinite number of elements, say $E_1, E_2, \dots \in \mathcal{C}_n$, then
\begin{align}
    P(\mathcal{X}) \geq P\left( \bigcup_{k=1}^\infty E_k \right) = \sum_{k=1}^\infty P(E_k) > \sum_{k=1}^\infty 2^{-n-1} \to \infty,
\end{align}
where the first equality holds because $P$ is an additive measure and the $E_n$ terms are disjoint, and the second inequality follows because $E_k \in \mathcal{C}_n$ so $P(E_k) > 2^{-n-1}$.
This results in a contradiction because $P(\mathcal{X}) = 1$, so each $\mathcal{C}_n$ must contain a finite number of terms.
Therefore, $\mathcal{C}$ is a countable union of finite sets, which is also countable.
This implies that the union of the elements of $\mathcal{C}$, namely $C = \cup_{C' \in \mathcal{C}} C'$ is a countable union of measurable sets and therefore also measurable.
Since $C$ is measurable, $H \setminus C$ is also measurable and we can rewrite the integral in \cref{eq:integral_form} as
{
\begin{align} 
    p_d &\geq \int_H dP ~ \frac{Q(B_d(x)) - \tau(B_d(x))}{P(B_d(x))} \\
    &= \int_{H \cap C} dP ~ \frac{Q(B_d(x)) - \tau(B_d(x))}{P(B_d(x))} + \int_{H \setminus C} dP ~ \frac{Q(B_d(x)) - \tau(B_d(x))}{P(B_d(x))} \label{eq:HC_sum} \\
    &\to 0
\end{align}
}
Since both terms above are non-negative and their sum converges to $0$, the terms must also individually converge to $0$.
Therefore, for the first term, we can write
\begin{equation}
    \label{eq:hc_term1}
    \lim_{d \to \infty} \int_{H \cap C} dP ~ \frac{Q(B_d(x)) - \tau(B_d(x))}{P(B_d(x))} = \liminf_{d \to \infty} \int_{H \cap C} dP ~ \frac{Q(B_d(x)) - \tau(B_d(x))}{P(B_d(x))} = 0.
\end{equation}
Similarly to $B_d$ defined in \cref{eq:defBd}, let us define $B : C \to \Sigma$ as
\begin{equation}
    B(x) = C' \in \mathcal{C} \text{ such that } x \in C'.
\end{equation}
Applying Fatou's lemma \cite[4.3.3, p. 131;][]{dudley2018real} to \cref{eq:hc_term1}, we obtain
\begin{align}
    \liminf_{d \to \infty} \int_{H \cap C} dP ~ \frac{Q(B_d(x)) - \tau(B_d(x))}{P(B_d(x))} &\geq \int_{H \cap C} dP ~ \liminf_{d \to \infty} \frac{Q(B_d(x)) - \tau(B_d(x))}{P(B_d(x))} \label{eq:fatou_1} \\
    &= \int_{H \cap C} dP ~ \frac{Q(B(x)) - \tau(B(x))}{P(B(x))} \label{eq:fatou_2} \\
    &= 0,
\end{align}
where from \cref{eq:fatou_1} to \cref{eq:fatou_2} we have used the fact that $P(B_d(x)) > 0$ whenever $x \in C$ and also that $B_1(x) \supseteq B_2(x) \supseteq \dots$.
Now we can re-write this integral as a sum, as follows.
Let the elements of $\mathcal{C}$, which we earlier showed is countable, be $C_1, C_2, \dots$ and write
\begin{align}
    \int_{H \cap C} dP ~ \frac{Q(B(x)) - \tau(B(x))}{P(B(x))} &= \sum_{n=1}^\infty \int_{H \cap C_n} dP ~ \frac{Q(B(x)) - \tau(B(x))}{P(B(x))} \\
    &= \sum_{n=1}^\infty \frac{P(H \cap C_n)}{P(C_n)} \left(Q(C_n) - \tau(C_n)\right) \\
    &= 0.
\end{align}
Now, from \cref{lem:Qtau}, we have
\begin{align}
    & \sum_{n=1}^\infty \frac{P(H \cap C_n)}{P(C_n)} \left(Q(C_n) - \tau(C_n)\right) = \sum_{n=1}^\infty \frac{P(H \cap C_n)}{P(C_n)} \left(Q(H \cap C_n) - \tau(H \cap C_n)\right) = 0,
\end{align}
which in turn implies that for each $n = 1, 2, \dots$, we have either $Q(H \cap C_n) - \tau(H \cap C_n) = 0$ or $P(H \cap C_n) = 0$.
However, the latter case also implies $Q(H \cap C_n) - \tau(H \cap C_n) = 0$ because $Q \ll P$, so $Q(H \cap C_n) - \tau(H \cap C_n) = 0$ holds for all $n$.
Therefore
\begin{equation}
    \tau(H \cap C) = \sum_{n=1}^\infty \tau(H \cap C_n) = \sum_{n=1}^\infty Q(H \cap C_n) = Q(H \cap C).
\end{equation}
Returning to the second term in the right hand of \cref{eq:HC_sum}, and again applying Fatou's lemma
\begin{align}
    \liminf_{d \to \infty} \int_{H \setminus C} dP ~ \frac{Q(B_d(x)) - \tau(B_d(x))}{P(B_d(x))} &\geq \int_{H \setminus C} dP ~ \liminf_{d \to \infty} \frac{Q(B_d(x)) - \tau(B_d(x))}{P(B_d(x))}. \label{eq:term2_fatou}
\end{align}
Now, since $Z$ has the nice-shrinking property from \cref{assumption:nice_shrinking}, we can apply a standard result from measure theory and integration \citet[given in Theorem 7.10, p. 140]{rudin86real}, to show that the following limit exists and the following equalities are satisfied
\begin{align}
    \lim_{d \to \infty} \frac{Q(B_d(x)) - \tau(B_d(x))}{P(B_d(x))} &= \lim_{d \to \infty} \frac{1}{P(B_d(x))} \int_{B_d} dP \left( \frac{dQ}{dP}(x) - \frac{d\tau}{dP}(x) \right) \\
    &= \frac{dQ}{dP}(x) - \frac{d\tau}{dP}(x) \label{eq:rudin}
\end{align}
Inserting \ref{eq:rudin} into \cref{eq:term2_fatou}, we obtain
\begin{align}
    \liminf_{d \to \infty} \int_{H \setminus C} dP ~ \frac{Q(B_d(x)) - \tau(B_d(x))}{P(B_d(x))} \geq \int_{H \setminus C} dP ~ \left(\frac{dQ}{dP}(x) - \frac{d\tau}{dP}(x)\right) = 0,
\end{align}
which in turn implies that
\begin{equation}
    \label{eq:qhcthc}
    \frac{dQ}{dP}(x) - \frac{d\tau}{dP}(x) = 0  ~~ P\text{-almost-everywhere on } H \setminus C,
\end{equation}
or equivalently that $Q(H \setminus C) = \tau(H \setminus C)$.
Combining this with the fact that $Q(\mathcal{X} \setminus H) = \tau(\mathcal{X} \setminus H)$ and our earlier result that $Q(H \cap C) = \tau(H \cap C)$, we have
\begin{equation*}
    ||Q - \tau||_{TV} = Q(\mathcal{X} \setminus H) - \tau(\mathcal{X} \setminus H) + Q(H \setminus C) - \tau(H \setminus C) + Q(H \cap C) - \tau(H \cap C) = 0,
\end{equation*}
which is equivalent to $w_d = ||Q - \tau_d||_{TV} \to 0$, that is the required result.
\end{proof}

% ==========================================
% Theorem: tau_d -> Q
% ==========================================

\begin{theorem*}[Correcness of GRC]
    If any one of the assumptions \ref{assumption:finite_ratio_mode}, \ref{assumption:single_branch} or \ref{assumption:nice_shrinking} holds, then
    \begin{equation}
    \label{eq:app:required}
        ||Q - \tau_d ||_{TV} \to 0 ~ \text{ as } ~ d \to \infty.
    \end{equation}
\end{theorem*}
\begin{proof}
If $p_d \to 0$, then $w_d \to 0$ by \cref{lemma:pdzero}.
If $P(H_d) \to 0$, then $w_d \to 0$ by \cref{lemma:hzero}.
Therefore suppose that $p_d \not \to 0$ and $P(H_d) \not \to 0$.
Then if any one of assumptions \ref{assumption:finite_ratio_mode}, \ref{assumption:single_branch} or \ref{assumption:nice_shrinking} holds, we can conclude from \cref{lemma:finite_ratio}, \ref{lemma:single_branch} or \ref{lemma:nicely_shrkinking} respectively, that $||Q - \tau_d ||_{TV} \to 0$.
\end{proof}

\newpage
\section{Optimality of GRCS}
\label{app:grcs_proofs}
\begin{algorithm}[H]
    \centering
    \caption{GRCS with arthmetic coding for the heap index.}
    \begin{algorithmic}[1]
        \Require Target $Q$, proposal $P$ over $\Reals$ with unimodal density ratio $r = dQ/dP$ with mode $\mu$.
        \State $d \gets 0, T_0 \gets 0, L_0 \gets 0$
        \State $I_0 \gets 1, S_1 \gets \Reals$
        \While{$\mathtt{True}$}
            \State $X_{I_d} \sim P|_{S_d}/P(S_d)$
            \State $U_{I_d} \sim \text{Uniform}(0, 1)$
            \State $\beta_{I_d} \gets \mathtt{clip}\left(P(S_d) \cdot \frac{r(X_{I_d}) - L_d}{1 - T_d}, 0, 1\right)$
            \Comment{$\mathtt{clip}(y, a, b) \defeq \max\{\min\{y, b\}, a\}$}
                \If{$U_{I_d} \leq \beta_{d+1}$}
                    \State \textbf{return} $X_{I_d}, I_d$
                \EndIf
            \If{$X_{I_d} > \mu$}
                \State $I_{d + 1} \gets 2I_d$
                \State $S_{d + 1} \gets S_d \cap (-\infty, X_{I_d})$
            \Else
                \State $I_{d + 1} \gets 2I_d + 1$
                \State $S_{d + 1} \gets S_d \cap (X_{I_d}, \infty)$
            \EndIf
            \State $L_{d + 1} \gets L_d + T_d / P(S_d)$
            \State $T_{d+1} \gets \Prob_{Y \sim Q}[r(Y) \geq L_{d + 1}] - L_{d + 1} \cdot \Prob_{Y \sim P}[r(Y) \geq L_{d + 1}]$
            \State $d \gets d + 1$
        \EndWhile
    \end{algorithmic}
    \label{alg:grcs}
\end{algorithm}
\noindent
In this section, we prove \Cref{thm:grcs:codelength,thm:grcs:runtime}.
We are only interested in continuous distributions over $\Reals$ with unimodal density ratio $dQ/dP$ for these theorems.
Hence, we begin by specializing \Cref{alg:grc} to this setting, shown in \Cref{alg:grcs}.
For simplicity, we also dispense with the abstraction of partitioning processes and show the bound update process directly.
Furthermore, we also provide an explicit form for the \texttt{AcceptProb} and \texttt{RuledOutMass} functions.
% \par
% GRCS can quickly converge to the accepted sample because it only samples from smaller and smaller regions of space.
% To formalize this notion, we introduce the notion of the \textit{contraction rate} of a partitioning process.
% \begin{definition}[Contraction rate of a partitioning process]
% Let $P$ be a probability measure and $Z$ a partitioning process over $\Reals$.
% Then, for $\epsilon \in [1/2, 1]$, we say that $Z$ contracts at least at rate $\epsilon$ with respect to $P$, if for all depths $d \geq 0$ for all sets $Z_{2^d}, Z_{2^d + 1}, \hdots, Z_{2^{d + 1} - 1}$ in the process at depth $d$ we have that
% \begin{align}
%     \max_{R \in [0:2^{d} - 1]}\Exp[P(Z_{2^d + R})] \leq \epsilon^d,
% \end{align}
% where the expectation is taken with respect to the randomness over the sets of the partitioning process.
% We call the smallest such $\epsilon$ the contraction rate of $Z$ with respect to $P$.
% \end{definition}
% Now, we show the following general result:
% \begin{theorem}
% Let $$
% \end{theorem}
\par
Before we move on to proving our proposed theorems, we first prove two useful results.
First, we bound the negative log $P$-mass of the bounds with which \Cref{alg:grcs} terminates.
\begin{lemma}
\label{lemma:negative_log_bound_size_bound}
Let $Q$ and $P$ be distributions over $\Reals$ with unimodal density ratio $r = dQ/dP$, given to \Cref{alg:grcs} as the target and proposal distribution as input, respectively.
Let $d \geq 0$ and let $X_{1:d} \defeq X_1, \hdots, X_d$ denote the samples simulated by \Cref{alg:grcs} up to step $d + 1$, where for $d = 0$ we define the empty list as $X_{1:0} = \emptyset$.
Let $S_d$ denote the bounds at step $d + 1$.
Then,
\begin{align}
-\sum_{j = 0}^d A_{j + 1}(\Reals, S_{0:d}) \cdot \log P(S_j) \leq \KLD{Q}{P} + \log e.
\end{align}
\end{lemma}
\begin{proof}
For brevity, we will write $A_d = A_d(\Reals, S_{0:d})$ and $T_d = T_d(\Reals, S_{0:d})$.
Furthermore, as in \Cref{alg:grcs}, we define 
\begin{align}
    L_d \defeq \sum_{j = 0}^{d - 1} \frac{1 - T_{j}}{P(S_j)} \quad \text{with} \quad L_0 = 0.
\end{align}
Note that $X_{1:d}$ is well-defined for all $d \geq 0$ since we could remove the return statement from the algorithm to simulate the bounds it would produce up to an arbitrary step $d$.
Now, note that by \Cref{prop:grs:TA} we have $\Prob[D = d \mid X_{1:d}] = A_{d + 1}(\Reals, S_{0:d})$.
Now, fix $d \geq 0$ and bounds $S_{0:d}$, and let $x \in \Reals$ be such that $\alpha_{d + 1}(x) > 0$ which holds whenever $r(x) \geq L_{d}$.
From this, for $d \geq 1$  we find
\begin{align}
    r(x) &\geq \sum_{j = 0}^{d - 1} \frac{1 - T_j}{P(S_j)} \\
    &\geq \frac{1 - T_{d - 1}}{P(S_{d - 1})},
\end{align}
where the second inequality follows from the fact that the $(1 - T_j) / P(S_j)$ terms are all positive.
taking logs, we get
\begin{align}
    \log r(x) - \log(1 - T_{d - 1}) \geq -\log P(S_{d - 1}).
    \label{eq:log_bound_size_ineq}
\end{align}
Now, we consider the expectation of interest:
\begin{align}
    \sum_{j = 0}^d -A_{j + 1} \cdot \log P(S_j)
    &= -\sum_{j = 0}^d \int_\Reals \alpha_{j + 1}(x)\log P(S_j)\, dx \\
    &\stackrel{\text{\cref{eq:log_bound_size_ineq}}}{\leq} \sum_{j = 0}^d \int_\Reals \alpha_{j + 1}(x)(\log(r(x)) - \log(1 - T_j))\,dx \\
    &\stackrel{(a)}{\leq} \int_\Reals \sum_{j = 0}^\infty\alpha_{j + 1}(x)\log r(x)\,dx + \sum_{j = 0}^\infty A_{j + 1} \log\frac{1}{1 - T_j} \\
    &\stackrel{(b)}{=} \int_\Reals q(x) \log r(x)\,dx + \sum_{j = 0}^\infty (T_{j + 1} - T_j) \log\frac{1}{1 - T_j} \\
    &=\KLD{Q}{P} + \sum_{j = 0}^\infty (T_{j + 1} - T_j) \log\frac{1}{1 - T_j} \\
    &\stackrel{(c)}{\leq} \KLD{Q}{P} \cdot \log 2 + \int_0^1\log\frac{1}{1 - t} \, dt \\
    &=\KLD{Q}{P} + \log e.
\end{align}
Inequality (a) holds because all terms are positive.
This is guaranteed by the fact that for $d \geq 1$, we have $L_d \geq 1$, hence $0 \leq \log L_d \leq r(x)$ whenever \Cref{eq:log_bound_size_ineq} holds.
Equality (b) follows by the correctness of GRC (\Cref{thm:correctness}), which implies that for all $x \in \Reals$ we have $\sum_{j=0}^\infty\alpha_d(x) = q(x)$, and inequality (c) follows from the facts that $0 \leq T_d \leq 1$ for all $d$ and that the summand in the second term forms a lower-Riemann sum approximation to $-\log(1 - t)$.
\end{proof}
\par
Second, we consider the contraction rate of the bounds $S_{0:d}$, considered by \Cref{alg:grcs}.
\begin{lemma}
\label{lemma:bound_size_exp_bound}
Let $Q$ and $P$ be distributions over $\Reals$ with unimodal density ratio $r = dQ/dP$, given to \Cref{alg:grcs} as the target and proposal distribution as input, respectively.
Assume $P$ has CDF $F_P$ and the mode of $r$ is at $\mu$.
Fix $d \geq 0$ and let $X_{1:d}$ be the samples considered by \Cref{alg:grcs} and $S_d$ the bounds at step $d + 1$. 
Then,
\begin{align}
    \Exp_{X_{1:d}}[P(S_d)] \leq \left(\frac{3}{4}\right)^d
\end{align}
\end{lemma}
\begin{proof}
We prove the claim by induction.
For $d = 0$ the claim holds trivially, since $S_0 = \Reals$, hence $P(S_0) = 1$.
Assume now that the claim holds for $d = k - 1$, and we prove the statement for $d = k$.
By the law of iterated expectations, we have
\begin{align}
    \Exp_{X_{1:k}}[P(S_{k})] = \Exp_{X_{1:k - 1}}[\Exp_{X_k \mid X_{1:k - 1}}[P(S_{k})]].
    \label{eq:bound_size_law_of_iterated_exp}
\end{align}
Let us now examine the inner expectation.
First, assume that $S_{k - 1} = (a, b)$ for some real numbers $a < b$ and define
$A = F_P(a), B = F_P(B), M = F_P(\mu)$ and $U = F_P(X_k)$.
Since $X_k \mid X_{1:k - 1} \sim P\vert_{S_{k - 1}}$, by the probability integral transform we have $U \sim \Unif(A, B)$, where $\Unif(A, B)$ denotes the uniform distribution on the interval $(A, B)$.
The two possible intervals from which \Cref{alg:grcs} will choose are $(a, X_k)$ and $(X_k, b)$, whose measures are $P((a, X_k)) = F_P(X_k) - F_P(a) = U - A$ and similarly $P((X_k, b)) = B - U$.
Then, $P(S_k) \leq \max\{U - A, B - U\}$, from which we obtain the bound
\begin{align}
\Exp_{X_k \mid X_{1:k - 1}}[P(S_{k})] \leq \Exp_{U}[\max\{U - A, B - U\}]
= \frac{3}{4}(B - A) = \frac{3}{4}P(S_{k - 1}).
\end{align}
Plugging this into \Cref{eq:bound_size_law_of_iterated_exp}, we get
\begin{align}
    \Exp_{X_{1:k}}[P(S_{k})] &\leq \frac{3}{4}\Exp_{X_{1:k - 1}}\left[P(S_{k - 1})\right] \\
    &\leq \frac{3}{4} \cdot \left(\frac{3}{4}\right)^{k - 1},
\end{align}
where the second inequality follows from the inductive hypothesis, which finishes the proof.
\end{proof}
\par
\textbf{The proof of \Cref{thm:grcs:runtime}:} We prove our bound on the runtime of \Cref{alg:grcs} first, as this will be necessary for the proof of the bound on the codelength.
First, let $D$ be the number of steps \Cref{alg:grcs} takes before it terminates minus $1$.
Then, we will show that
\begin{align}
    \Exp[D] \leq \frac{1}{\log(4 / 3)}\KLD{Q}{P} + 4
\end{align}
We tackle this directly. 
Hence, let
\begin{align}
    \Exp_D[D] &= \lim_{d \to \infty} \Exp_{X_{1:j}}\left[\sum_{j = 1}^d j \cdot A_{j + 1}\right] \\
    &=\lim_{d \to \infty} \Exp_{X_{1:j}}\left[\sum_{j = 1}^d \frac{-j}{\log P(S_j)} \cdot -A_{j + 1}\log P(S_j)\right] \\
    &\leq \lim_{d \to \infty} \Exp_{X_{1:j}}\left[\max_{j \in [1:d]}\left\{\frac{-j}{\log P(S_j)}\right\} \cdot \sum_{j = 1}^d  -A_{j + 1}\log P(S_j)\right] \\
    &\stackrel{\text{\cref{lemma:negative_log_bound_size_bound}}}{\leq} \left(\KLD{Q}{P} + \log e\right) \cdot \lim_{d \to \infty} \Exp_{X_{1:j}}\left[\max_{j \in [1:d]}\left\{\frac{-j}{\log P(S_j)}\right\}\right].
\end{align}
To finish the proof, we will now bound the term involving the limit.
To do this, note, that for any finite collection of reals $F$, we have $\max_{x \in F}\{x\} = -\min_{x \in F}\{-x\}$, and that for a finite collection of real-valued random variables $\hat{F}$ we have $\Exp[\min_{\rvx \in \hat{F}}\{\rvx\}] \leq \min_{\rvx \in \hat{F}}\{\Exp[\rvx]\}$.
Now, we have
\begin{align}
    \lim_{d \to \infty} \Exp_{X_{1:j}}\left[\max_{j \in [1:d]}\left\{\frac{-j}{\log P(S_j)}\right\}\right] &= \lim_{d \to \infty} -\Exp_{X_{1:j}}\left[\min_{j \in [1:d]}\left\{\frac{j}{\log P(S_j)}\right\}\right] \\
    &\leq \lim_{d \to \infty} \left(-\min_{j \in [1:d]}\left\{\Exp_{X_{1:j}}\left[\frac{j}{\log P(S_j)}\right]\right\}\right) \\
    &\stackrel{\text{(a)}}{\leq} \lim_{d \to \infty} \left(-\min_{j \in [1:d]}\left\{\frac{j}{\log \Exp_{X_{1:j}}\left[P(S_j)\right]}\right\}\right) \\
    &\stackrel{\text{\cref{lemma:bound_size_exp_bound}}}{\leq} \lim_{d \to \infty} \left(-\min_{j \in [1:d]}\left\{\frac{-j}{j\log (4/3)}\right\}\right) \\
    &= \lim_{d \to \infty} \left(\max_{j \in [1:d]}\left\{\frac{1}{\log (4/3)}\right\}\right) \\
    &= \frac{1}{\log(4/3)}
\end{align}
Inequality (a) follows from Jensen's inequality.
Finally, plugging this back into the previous equation, we get
\begin{align}
    \Exp[D] \leq \frac{\KLD{Q}{P} + \log e}{\log 4/3} \leq \frac{\KLD{Q}{P}}{\log 4/3} + 4
\end{align}
\par
\textbf{Proof of \Cref{thm:grcs:codelength}:}
For the codelength result, we need to encode the length of the search path and the search path itself.
More formally, since the returned sample $X$ is a function of the partition process $Z$, the search path length $D$ and search path $S_{0:D}$, we have
\begin{align}
    \Ent[X \mid Z] \leq \Ent[D, S_{0:D}] = \Ent[D] + \Ent[S_{0:D} \mid D].
\end{align}
we can encode $D$ using Elias $\gamma$-coding, from which we get
\begin{align}
    \Ent[D] &\leq \Exp_D[2\log (D + 1)] + 1 \\
    &\leq 2\log(\Exp[D] + 1) + 1 \\
    &\leq 2\log\left(\frac{\KLD{Q}{P} + \log e}{\log (4/3)} + 1\right) + 1 \\
    &\leq 2\log\left(\KLD{Q}{P} + \log e + \log(4/3) \right) + 1 - 2\log\left(\log (4/3)\right) \\
    &\leq 2\log\left(\KLD{Q}{P} + 1 \right) + 1 - 2\log\left(\log (4/3)\right) + 2\log(\log e + \log(4/3))\\
    &\leq 2\log\left(\KLD{Q}{P} + 1 \right) + 6.
\end{align}
Given the search path length $D$, we can use arithmetic coding (AC) to encode the sequence of bounds $S_{0:D}$ using $-\log P(S_D) + 2$ bits (assuming infinite precision AC).
Hence, we have that the average coding cost is upper bounded by
\begin{align}
    \Ent[S_{0:D} \mid D] \leq \Exp_D[-\log P(S_D)] + 2 \stackrel{\text{\cref{lemma:negative_log_bound_size_bound}}}{\leq} \KLD{Q}{P} + 5.
\end{align}
Putting everything together, we find
\begin{align}
    \Ent[D, S_{0:D}] \leq \KLD{Q}{P} + 2\log(\KLD{Q}{P} + 1) + 11,
\end{align}
as required.

\section{Additional experiments with depth-limited GRC}

In this section we show the results of some experiments comparing the approximation bias of depth limited GRCD, to that of depth limited \adstar, following the setup of \cite{flamich2022fast}.
Limiting the depth of each algorithm introduces bias in the resulting samples, as these are not guaranteed to be distributed from the target distribution $Q$, but rather from a different distribution $\smash{\hat{Q}}$.
\Cref{fig:approx_rec} quantifies the effect of limiting the depth on the bias of the resulting samples.

In our experiment we take $Q$ and $P$ to be Gaussian and we fix $\KLD{Q}{P} = 3$ (bits), and consider three different settings of $\infD{Q}{P} = 5, 7$ or $9$ (bits), corresponding to each of the panes in \cref{fig:approx_rec}.
For each such setting, we set the depth limit of each of the two algorithms to $D_{\max} = \KLD{Q}{P} + d$ bits, and refer to $d$ as the \textit{number of additional bits}.
We then vary the number of additional bits allowed for each algorithm, and estimate the bias of the resulting samples by evaluating the KL divergence between the empirical and the exact target distribution, that is $\smash{\KLD{\hat{Q}}{Q}}$.
To estimate this bias, we follow the method of \cite{perez2008kullback}.
For each datapoint shown we draw 200 samples $X \sim \hat{Q}$ and use these to estimate $\smash{\KLD{\hat{Q}}{Q}}$.
We then repeat this for 10 different random seeds, reporting the mean bias and standard error in the bias, across these 10 seeds.

Generally we find that the bias of GRCD is higher than that of \adstar.
This is likely because \adstar is implicitly performing importance sampling over a set of $2^{D_{\max} + d} - 1$ samples, and returning the one with the highest importance weight.
By contrast, GRCD is running rejection sampling up to a maximum of $D_{\max} + d$ steps, returning its last sample if it has not terminated by its $(D_{\max} + d)^{\text{th}}$ step.
While it might be possible to improve the bias of depth limited GRCD by considering an alternative way of choosing which sample to return, using for example an importance weighting criterion, we do not examine this here and leave this possibility for future work.

\begin{figure}[h!]
    \centering
    \begin{subfigure}[b]{0.99\textwidth}
        \includegraphics[width=\textwidth]{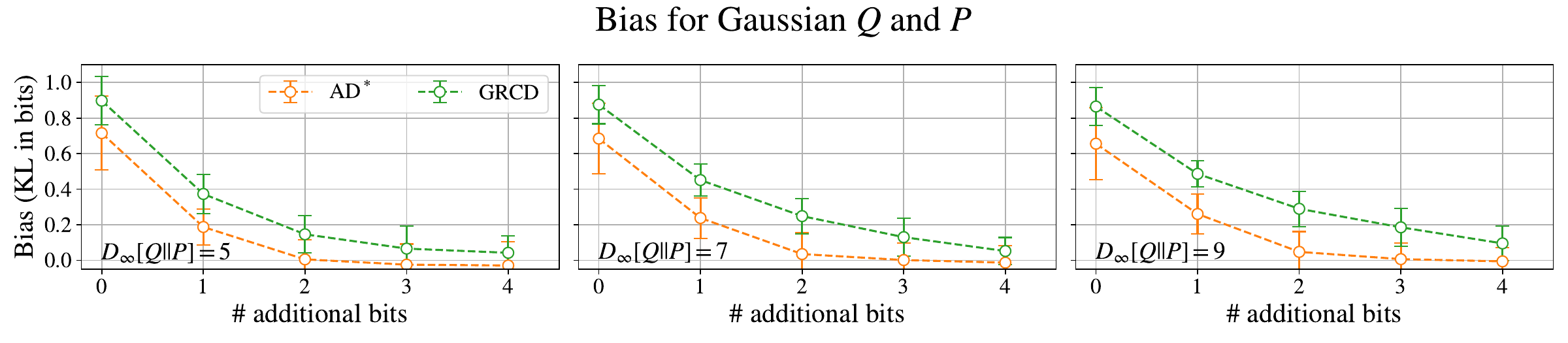}
    \end{subfigure}
    \caption{Bias of depth-limited \adstar and GRCD, as a function of the number of additional bit budget given to each algorithm.
    See text above for discussion.}
    \label{fig:approx_rec}
\vspace{-4mm}
\end{figure}

\end{document}